\documentclass[sigplan, screen]{acmart}  % {{{1
\usepackage{mathtools, enumitem, booktabs, float, caption, tikz}
\newcommand{\NN}{\mathbb{N}}
\newcommand{\ZZ}{\mathbb{Z}}

\newcommand{\CC}{\mathbb{C}}
\newcommand{\m}[1]{\mathbb{#1}}      % models
\newcommand{\alg}[1]{\m{#1}}         % algebras
\newcommand{\cl}[1]{\mathcal{#1}}    % classes
\renewcommand{\AA}{\alg{A}}

\theoremstyle{acmplain}
  \newtheorem{thm}{Theorem}[section]  \newtheorem*{thm*}{Theorem}
        \newtheorem*{claim*}{Claim}
    \newtheorem*{conj*}{Conjecture}
  \newtheorem{cor}[thm]{Corollary}    \newtheorem*{cor*}{Corollary}
  \newtheorem{lem}[thm]{Lemma}        \newtheorem*{lem*}{Lemma}
  \newtheorem{prop}[thm]{Proposition} \newtheorem*{prop*}{Proposition}
\theoremstyle{acmdefinition}
  \newtheorem{defn}[thm]{Definition}  \newtheorem*{defn*}{Definition}
  \newtheorem{obs}[thm]{Observation}  \newtheorem*{obs*}{Observation}
  \newtheorem*{question*}{Question}

\newenvironment{claimproof}
  {
    \begin{proof}[Proof of claim]
    
  } {
    \end{proof}
  }

\DeclareMathOperator{\Cg}{Cg}
\DeclareMathOperator{\Clo}{Clo}
\DeclareMathOperator{\Con}{Con}

\DeclareMathOperator{\Sg}{Sg}

\newcommand{\stack}[2]{\genfrac{}{}{0pt}{}{#1}{#2}}

\usetikzlibrary{calc, matrix, positioning}

\numberwithin{equation}{section}  % number equations within sections
\renewcommand{\epsilon}{\varepsilon}
\renewcommand{\phi}{\varphi}

\setlist[enumerate]{itemsep=0.5em}
\setlist[itemize]{itemsep=0.5em}

\usepackage{array}

\theoremstyle{acmplain} \newtheorem{innercustomthm}{Theorem}
\newenvironment{customthm}[1]
  {\renewcommand\theinnercustomthm{#1}\innercustomthm}
  {\endinnercustomthm}
\theoremstyle{acmplain} \newtheorem{innercustomcor}{Corollary}
\newenvironment{customcor}[1]
  {\renewcommand\theinnercustomcor{#1}\innercustomcor}
  {\endinnercustomthm}

\DeclareMathOperator{\HKP}{HKP}
\DeclareMathOperator{\HSP}{HSP}
\DeclareMathOperator{\Simon}{\mathfrak{s}}
\DeclareMathOperator{\proj}{proj}
\newcommand{\Cong}[1]{\mathbf{#1}}
\newcommand{\comp}[1]{\textsf{#1}}
\renewcommand{\O}{\mathcal{O}}

\newcommand{\Clone}[1]{\mathcal{#1}}
\newcommand{\CloneBot}{\boldsymbol{\bot}}
\newcommand{\CloneTop}{\boldsymbol{\top}}
\newcommand{\CloneVee}{\boldsymbol{\bigvee}}
\newcommand{\CloneWedge}{\boldsymbol{\bigwedge}}
\newcommand{\Aiff}{\leftrightarrow}
\DeclareMathOperator{\maj}{maj}

\newcommand{\bra}[1]{ \left< #1 \right| }
\newcommand{\ket}[1]{ \left| #1 \right> }

\newcommand{\B}{\mathfrak{B}}
\DeclareMathOperator{\Cspan}{\CC-span}
\DeclareMathOperator{\Tr}{Tr}
\DeclareMathOperator{\Prob}{P}
\tikzset{gate-c/.style = {draw, circle, inner sep=0.25em}}
\tikzset{gate-r/.style = {draw, rounded corners, rectangle}}
\tikzset{control/.style = {draw, fill, circle, inner sep=0.1em}}
\newcommand{\tikzMeasure}{
  \begin{tikzpicture}[scale=0.2, baseline=0pt]
    \draw (0,0) arc (180-20:20:1);
    \draw[->] (current bounding box.north|-current bounding box.south) -- ++(60:1.5);
  \end{tikzpicture} }

\tikzset{clone/.style={draw, circle, inner sep=0.125em}}
\tikzset{cloneLabel/.style={preaction={fill,white}, inner sep=0.05em, outer sep=0.33em}}
\tikzset{curveR/.style={out=90-30,in=-90+30}}
\tikzset{curveL/.style={out=90+30,in=-90-30}}
\tikzset{class_and_quant/.style={}}
\tikzset{superpoly/.style={very thick}}
\tikzset{no_poly_node/.style={black!70, thick}}
\tikzset{no_poly_line/.style={densely dotted, thick}}
\tikzset{misc/.style={black!30}}
\copyrightyear{2020}
\acmYear{2020}
\setcopyright{acmlicensed}
\acmConference[LICS '20]
  {Proceedings of the 35th Annual ACM/IEEE Symposium on Logic in Computer Science (LICS)}
  {July 8--11, 2020}
  {Saarbr\"ucken, Germany}
\acmBooktitle{Proceedings of the 35th Annual ACM/IEEE Symposium on Logic in
  Computer Science (LICS '20), July 8--11, 2020, Saarbr\"ucken, Germany}
\acmPrice{15.00}
\acmDOI{10.1145/3373718.3394764}
\acmISBN{978-1-4503-7104-9/20/07}

\keywords{quantum computation, hidden subgroup problem, hidden kernel problem,
  universal algebra, clone theory}

\ccsdesc[500]{Theory of computation~Quantum query complexity}
\ccsdesc[500]{Theory of computation~Finite Model Theory}
\ccsdesc[300]{Theory of computation~Quantum complexity theory}
\ccsdesc[300]{Theory of computation~Problems, reductions and completeness}
\begin{document}
% title and abstract {{{1
\author{Matthew Moore}
\affiliation{\institution{The University of Kansas}}
\email{matthew.moore@ku.edu}

\author{Taylor Walenczyk}
\affiliation{\institution{The University of Kansas}}
\email{tawalenczyk@ku.edu}

\title{The Hidden Subgroup Problem for Universal Algebras}
\date{\today}

\begin{abstract}
The Hidden Subgroup Problem (HSP) is a computational problem which includes as
special cases integer factorization, the discrete logarithm problem, graph
isomorphism, and the shortest vector problem. The celebrated polynomial-time
quantum algorithms for factorization and the discrete logarithm are restricted
versions of a generic polynomial-time quantum solution to the HSP for
\emph{abelian} groups, but despite focused research no polynomial-time solution
for general groups has yet been found. We propose a generalization of the HSP to
include \emph{arbitrary} algebraic structures and analyze this new problem on
powers of 2-element algebras. We prove a complete classification of every such
power as quantum tractable (i.e.\ polynomial-time), classically tractable,
quantum intractable, or classically intractable. In particular, we identify a
class of algebras for which the generalized HSP exhibits super-polynomial
speedup on a quantum computer compared to a classical one.
\end{abstract} 

\maketitle
%----------------------------------------------------------------------------}}}1

\section{Introduction} \label{sec:intro}  % {{{1
Quantum algorithms enable new methods of computation through the exploitation of
the natural phenomena associated with quantum systems. Sometimes, this makes
possible efficient (i.e.\ polynomial-time) solutions to computational problems
for which no efficient classical algorithm exists. While there is no proof that
quantum computation is super-polynomially faster than classical methods, there
is a growing body of evidence in support of such a claim.

Perhaps the most famous quantum algorithm exhibiting super-polynomial speedup
over classical ones is Shor's algorithm for integer
factorization~\cite{Shor_Alg}. Other examples include Shor's solution to the
discrete logarithm problem~\cite{Shor_Alg}, Simon's algorithm~\cite{Simon_Alg},
the unit and class group algorithms~\cite{Hallgren_QuantumPellPID}, and Pell's
equation and the principal ideal algorithm~\cite{Hallgren_QuantumPellPID}
(see~\cite{NielsenChuang_QCQI} for more details). Surprisingly, these
superficially diverse problems are all special cases of a unifying problem
called the \emph{Hidden Subgroup Problem}.

Given a group $\alg{G}$, a set $X$, and a function $f: G \to X$, the function
$f$ \emph{hides} a subgroup $\alg{D}\leq \alg{G}$ if $f(g) = f(h)$ if and only
if $gD = hD$ (i.e.\ $f$ is constant exactly on the cosets of $\alg{D}$ in
$\alg{G}$). We follow the convention in Universal Algebra of distinguishing
between an algebraic structure (e.g.\ a group) and its underlying set, using
$\m{G}$ for the former and $G$ for the latter.
\begin{center} \begin{tabular}{rl}
  \multicolumn{2}{l}{\textbf{Hidden Subgroup Problem (HSP)}} \\ \midrule
  \texttt{Input:} & group $\alg{G}$, \\
                  & function $f: G \to X$ hiding some subgroup \\[0.25em]
  \texttt{Task:}  & determine the subgroup $\alg{D}\leq \alg{G}$ that $f$ hides
\end{tabular} \end{center}
The group $\alg{G}$ may be specified either by providing a multiplication table
or by fixing a family $(\alg{G}_n)_{n\in \NN}$ of groups (e.g.\ the symmetric
groups $\alg{S}_n$) and providing an index $n$. Formally, these two kinds of
specification will determine different input sizes to the problem. In any case,
however, we \emph{always} consider the input size to be $\lg |G|$. The function
$f$ is given as an oracle (i.e.\ a black-box function). Typically, $\alg{D}$ is
`determined' by generators, but any other scheme which uniquely identifies
$\alg{D}$ amongst all other subgroups is permitted. 

When $\alg{G}$ is \emph{abelian}, $\HSP(\alg{G})$ includes all of the problems
in the previous paragraph, and a general solution in this setting is due to
Kitaev~\cite{Kitaev_AbStab}, though it includes many of the aforementioned
previously known algorithms. The case for non-abelian $\alg{G}$ has proven to be
more challenging.

The HSP for \emph{non-abelian} $\alg{G}$ has been solved in some quite specific
cases~\cite{RoettelerBoth_HSPNonAb, IvanyosMagniezSantha_HSPNonAb,
MooreEtal_HSPAffine, Gavinsky_HSPPolynearham, BaconChildsVanDam_HSPSemidirect,
IvanyosSanselmeSantha_HSPExtraspecial, IvanyosSanselmeSantha_HSPNil2,
DenneyMooreRussel_HSPPSL2}, but it remains open in general. The HSP for the
symmetric groups and the dihedral groups correspond to the \emph{graph
isomorphism} and \emph{shortest vector} problems, respectively, and so these
cases are of particular interest. Despite focused research over the years, no
general polynomial-time quantum algorithm has been found.

We propose a generalization of the Hidden Subgroup Problem to arbitrary
algebraic structures called \emph{universal algebras}, and call the generalized
problem the \emph{Hidden Kernel Problem (HKP)}. Perhaps the simplest non-trivial
instance of the HSP is when $\alg{G} = \ZZ_2^n$, and this is known as Simon's
Problem. Along the same lines, we consider the Hidden Kernel Problem for powers
of 2-element structures, and give a complete classification of which structures
admit a polynomial-time quantum solution, a polynomial-time classical solution,
no polynomial-time classical solution, and no polynomial-time quantum solution.
In particular, we fully classify all powers of 2-element universal algebras for
which a quantum algorithm achieves super-polynomial speedup when compared to a
classical algorithm.

We begin with some necessary background in Universal Algebra and Clone theory,
as well as a brief discussion of Quantum Computation in
Section~\ref{sec:background}. In Section~\ref{sec:HKP} we develop and state the
Hidden Kernel Problem and outline the plan of attack for analyzing the HKP on
powers of 2-element universal algebras (of which there are infinitely many). We
determine which of these algebras admit polynomial-time quantum and classical
solutions to their HKP in Section~\ref{sec:hkp_positive}, and in
Section~\ref{sec:hkp_negative} we determine which have an exponential
lower-bound for classical or quantum solutions. Taken together, these results
provide a complete classification of the quantum as well as classical
computational complexity of the HKP on powers of 2-element universal algebras.
Section~\ref{sec:concl} concludes with an overview of the results as well as a
discussion of some open questions about the Hidden Kernel Problem.
%----------------------------------------------------------------------------}}}1

\section{Background} \label{sec:background}  % {{{1
Universal Algebra seeks to include and extend classical algebraic results (for
instance, from group theory and ring theory). In bringing together diverse areas
of research, Universal Algebra provides a common scaffolding upon which to build
general techniques. ``What looks messy and complicated in a particular framework
may turn out to be simple and obvious in the proper general one,'' as
Smith~\cite{Smith_MaltsevVars} describes it.

The subsections below give a brief overview of the topics needed for the later
sections, and though self-contained, they are necessarily brief. The interested
reader is encouraged to consult~\cite{MMT_ALVBook}
and~\cite{Taylor_CloneTheory}, which are both good references for the algebraic
topics. The section concludes with a short overview of Quantum Computation, for
which~\cite{NielsenChuang_QCQI} is a more extensive reference.

\subsection{Universal Algebra} \label{subsec:algebra} % {{{2
Given a set $A$, an \emph{operation} on $A$ is a function $f: A^n \to A$, where
$n\in \NN$ is the arity of $f$. The set $A$ together with some operations
$(f_i)_{i\in I}$ is a \emph{universal algebra}, commonly denoted $\AA$. We
specify the set and the operations by writing
\[
  \AA = \left< A \;;\; (f_i)_{i\in I} \right>
  \qquad\text{or}\qquad
  \AA = \left< A \;;\; f_1, \dots, f_k \right>,
\]
depending on whether $I$ is finite. A non-empty subset $B\subseteq A$ which is
closed under all operations $f_i$ of $\AA$ is called a \emph{subalgebra},
written $\m{B}\leq \AA$. If $C\subseteq A$ is non-empty, then we define the
\emph{subalgebra generated by C} to be the smallest subalgebra containing $C$,
written $\Sg(C)$. These general definitions of algebras and subalgebras include
most of the classical algebraic structures (groups, rings, boolean algebras,
etc.) and allow for the unification of a host of results.

The appropriate generalization of a normal subgroup for a universal algebra
is a \emph{congruence}. Precisely, a congruence $\theta$ of an algebra $\AA$
is a binary equivalence relation on the set $A$ which is compatible with the
operations of $\AA$, meaning
\[
  a_1 \;\theta\; b_1, \dots, a_k \;\theta\; b_k
  \quad \Rightarrow \quad
  f(a_1, \dots, a_k) \;\theta\; f(b_1, \dots, b_k)
\]
for all operations $f$ of $\AA$ (infix notation is used here for binary
relations). As with subalgebras, given a set $D\subseteq A^2$, we define the
\emph{congruence generated by $D$} to be the smallest congruence containing $D$,
written $\Cg(D)$. Similar to normal subgroups, the quotient structure
$\AA/\theta$ is well-defined, and general versions of the classic isomorphism
theorems for groups and rings hold true in this context.

Two algebras $\AA$ and $\m{B}$ are said to be \emph{similar} if there is a
bijective correspondence between the $k$-ary operations of $\AA$ and $\m{B}$ for
each $k$. In this case, $f^{\AA}$ is written for the operation of $\AA$
corresponding to the operation $f^{\m{B}}$ of $\m{B}$. When there is no chance
for confusion, however, these superscripts are omitted and we simply use the
same symbol $f$ for both $f^{\AA}$ and $f^{\alg{B}}$. A \emph{homomorphism}
between similar structures $\AA$ and $\m{B}$ is a function $\phi: A\to B$ such
that for each $k$-ary operation $f^{\AA}$ of $\AA$ and all $a_1, \dots, a_k\in
A$,
\[
  \phi\big( f^{\AA}(a_1, \dots, a_k) \big)
  = f^{\m{B}}\big( \phi(a_1), \dots ,\phi(a_k) \big) 
\]
(i.e.\ $\phi$ is compatible with the operations of $\AA$ and $\m{B}$). This is
indicated by writing $\phi : \AA \to \m{B}$. Similar to groups and rings,
congruences are precisely the \emph{kernels} of homomorphisms
\[
  \ker(\phi)
  = \big\{ (a,b)\in A^2 \mid \phi(a) = \phi(b) \big\},
\]
but in contrast to groups the kernel is a binary relation.

Every algebra $\AA$ has at least two (possibly non-distinct) congruences --- the
trivial congruence $\Cong0$ and the universal congruence $\Cong1$,
\begin{align*}
  \Cong0
  &= \big\{ (a,a)\in A^2 \mid a\in A \big\} & \text{and} \\
  \Cong1
  &= A^2 = \big\{ (a,b)\in A^2 \mid a,b\in A \big\},
\end{align*}
As with groups, if these are the only congruences of $\AA$, then we call $\AA$
\emph{simple}.

The set of congruences of $\AA$ is denoted $\Con(\AA)$, and can be ordered by
inclusion, in which case $\Cong1$ and $\Cong0$ are the maximal and minimal
elements, respectively. More than this, $\Con(\AA)$ has the structure of a
lattice, where the operations of $\wedge$ and $\vee$ are defined
\[
  \theta \wedge \psi 
    \coloneqq \theta \cap \psi
  \qquad\text{and}\qquad
  \theta \vee \psi
    \coloneqq \Cg\big( \theta \cup \psi \big)
\]
for $\theta, \psi\in \Con(\AA)$. Perhaps surprisingly, different kinds of
lattice-theoretic structure in $\Con(\AA)$ are in correspondence with algebraic
structure in $\AA$, and this correspondence is a central area of study in the
field. Of particular interest for us is when $\Con(\AA)$ is a \emph{distributive
lattice}, meaning
\[
  \alpha \wedge \big( \beta \vee \gamma \big)
  = \big( \alpha \wedge \beta \big) \vee \big( \alpha \wedge \gamma \big)
\]
for all $\alpha, \beta, \gamma \in \Con(\AA)$. In this case, we say that $\AA$
is \emph{congruence distributive}.
%---------------------------------------------------------------------------}}}2

\subsection{Clones} \label{subsec:clones} % {{{2
Given a set $F$ of operations on a common domain $A$, new operations may be
produced by composing operations. Formally, if $f\in F$ is $n$-ary and
$g_i\in F$ is $k_i$-ary for $i\in [n]$, then
\begin{multline*}
  h(x_{11}, \dots, x_{nk_n}) \\
  \coloneqq f\big( g_1(x_{11}, \dots, x_{1k_i}), \dots, g_n(x_{n1}, \dots, x_{nk_n}) \big)
\end{multline*}
is a $(\sum k_i)$-ary operation of $A$.

New operations of $A$ may also be produced by variable manipulations. Given
$n$-ary operation $f\in F$ and any function $\sigma: [n] \to [m]$,
\[
  h(x_1, \dots, x_m)
  \coloneqq f\big(x_{\sigma(1)}, \dots, x_{\sigma(n)}\big)
\]
is an $m$-ary operation of $A$. Less formally, given $n$-ary $f\in F$ we may
define operations of equal arity by permuting variables, operations of smaller
arity by identifying variables, and operations of larger arity by introducing
extraneous variables.

A \emph{clone} over a domain $A$ is a set of operations of $A$ which is closed
under composition and variable manipulations (as detailed above). The
\emph{clone generated by} a set of operations $F$, written $\Clo(F)$ is the
smallest clone containing $F$. Likewise, the \emph{clone of term operations} of
the algebra $\AA$, written $\Clo(\AA)$, is the clone generated by the operations
of $\AA$.

Given a fixed domain $A$, the set of clones over $A$ can be ordered by
inclusion. As with the set of congruences of an algebra, the set of clones over
$A$ forms a lattice with operations $\wedge$ and $\vee$ defined
\[
  \Clone{A} \wedge \Clone{B} \coloneqq \Clone{A}\cap \Clone{B}
  \qquad\text{and}\qquad
  \Clone{A} \vee \Clone{B} \coloneqq \Clo\big( \Clone{A}\cup \Clone{B} \big)
\]
for clones $\Clone{A}$ and $\Clone{B}$ over the domain $A$. We extend the
ordering on clones to algebras $\alg{B}$ and $\alg{D}$ over the common domain
$A$ by writing $\alg{B} \preceq \alg{D}$ if and only if $\Clo(\alg{B}) \subseteq
\Clo(\alg{D})$. If $\cl{C}$ is a clone over the domain $A$, then we will
slightly abuse this notation by writing $\cl{C} \preceq \alg{B}$ if
$\cl{C}\subseteq \Clo(\alg{B})$, and similarly for $\alg{B} \preceq \cl{C}$.
%---------------------------------------------------------------------------}}}2

\subsection{Quantum Computation} \label{subsec:quantum_algs} % {{{2
Qubits are the information-theoretic foundation of quantum computation. A qubit
can have a \emph{pure} state of either $0$ or $1$, but may also be in a
\emph{superposition} of $0$ and $1$ states. Formally, this is represented by a
2-dimensional normed complex vector space (i.e.\ a Hilbert space),
\[
  \B
  \coloneqq \Cspan\big\{ \ket{0}, \ket{1}\big\}.
\]
We adopt the notation of Dirac~\cite{Dirac_Ket} where $\ket{\alpha}$ denotes a
column vector and $\bra{\alpha} \coloneqq \ket{\alpha}^{\dagger}$ denotes a row
vector (``$\dagger$'' represents the Hermitian adjoint).

A system of $k$ qubits is represented by the tensor power of $\B$, written
$\B^{\otimes k} \coloneqq \B\otimes \dots \otimes \B$. This vector space has
dimension $2^k$ with basis $\big\{ \ket{x_1} \otimes \dots \otimes \ket{x_k}
\mid x_i\in \{0,1\} \big\}$. Using the notation $\ket{x_1x_2} \coloneqq
\ket{x_1}\otimes \ket{x_1}$ gives us
\[
  \B^{\otimes k}
  = \Cspan\big\{ \ket{x_1 \cdots x_k } \mid x_i \in \{0,1\} \big\}.
\]

A \emph{quantum state} is a vector in $\B^{\otimes k}$,
\[
  \ket{\alpha}
  = \sum_{x_i\in \{0,1\}} \lambda_{x_1 \cdots x_k} \ket{x_1 \cdots x_k}
  \quad\text{such that}\quad
  \|\alpha\| = 1.
\]
The condition that $\|\alpha\| = 1$ is equivalent to $\sum |\lambda_{x_1 \cdots
x_k}|^2 = 1$, and so $\big( |\lambda_{x_1 \cdots x_k}|^2 \big)_{x_i\in \{0,1\}}$ is
regarded as a \emph{probability distribution}: when the quantum state
$\ket{\alpha}$ is measured, the pure state $\ket{y_1 \cdots y_k}$ will be
observed with probability $|\lambda_{y_1 \cdots y_k}|^2$.

The evolution over time of a quantum state $\ket{\alpha}$ is represented by the
action of a unitary operator $U$ on $\ket{\alpha}$. A given unitary operator can
be approximated using products and tensor products of operators from the
\emph{standard quantum gate set}. A \emph{quantum circuit} for the operator $U$
is this decomposition, typically given as a diagram. Figure~\ref{fig:simon} is
an example of such a diagram and Definition~\ref{defn:simon_circuit} is the
inline form of the same circuit.

One particular quantum gate which we will make use of is the Hadamard gate $H$,
defined by its action on the basis vectors of $\B$ by
\[
  H\ket{0}
  = \frac{1}{\sqrt{2}} \big( \ket{0} + \ket{1} \big),
  \qquad
  H\ket{1}
    = \frac{1}{\sqrt{2}} \big( \ket{0} - \ket{1} \big).
\]
A convenient formula for the action of $H^{\otimes n}$ on an arbitrary
$n$-qubit state is
\[
  H^{\otimes n}\ket{x}
    = \frac{1}{2^{n/2}} \sum_{a\in \{0,1\}^n} (-1)^{a\cdot x} \ket{a},
\]
where $a\cdot x$ is the dot product of $a$ and $x$ modulo $2$ (regarded as
$\ZZ_2$-vectors).
%---------------------------------------------------------------------------}}}2
%---------------------------------------------------------------------------}}}2
%----------------------------------------------------------------------------}}}1

\section{The Hidden Kernel Problem} \label{sec:HKP}  % {{{1
% fig:posts_lattice {{{2
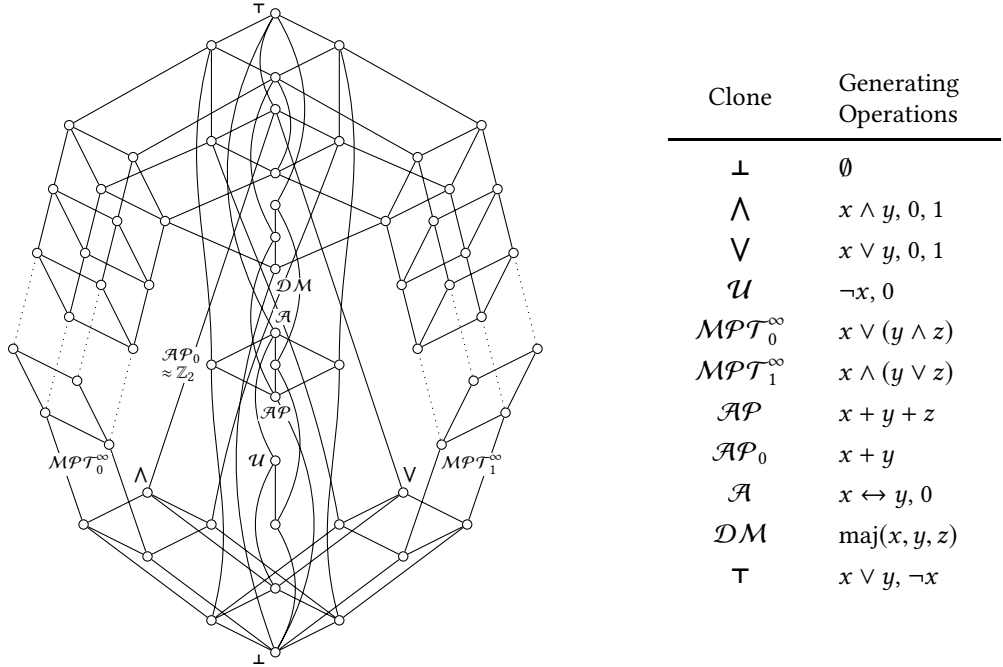
\begin{figure*} \centering
% Post's Lattice {{{3
\begin{tikzpicture}[font=\tiny, yscale=0.85, xscale=0.85, baseline=(center.base)]
  % central section nodes {{{
  \path (0, 0)    node[clone] (F) {}
       +(-1, 0.5) node[clone] (UP0) {}
       +(1, 0.5)  node[clone] (UP1) {}
       ++(0,1)    node[clone] (UM) {}
       ++(0,1)    node[clone] (UD) {}
       ++(0,1)    node[clone] (U) {}
       ++(0,1)    node[clone] (AP) {}
        +(-1,0.5) node[clone] (AP0) {}
        +(0,0.5)  node[clone] (AD) {}
        +(1,0.5)  node[clone] (AP1) {}
       ++(0,1)    node[clone] (A) {}
       ++(0,1)    node[clone] (DM) {}
       ++(0,0.5)  node[clone] (DP) {}
       ++(0,0.5)  node[clone] (D) {}
       ++(0,0.5)  node[clone] (MP) {}
        +(-1,0.5) node[clone] (MP0) {}
        +(1,0.5)  node[clone] (MP1) {}
       ++(0,1)    node[clone] (M) {}
       ++(0,0.5)  node[clone] (P) {}
        +(-1,0.5) node[clone] (P0) {}
        +(1,0.5)  node[clone] (P1) {}
       ++(0,1)    node[clone] (T) {}; % }}}
  % left wing nodes {{{
  \path (UP0)
      ++(0,1.5)     node[clone] (MeetP1) {}
      ++(-1,-0.5)   node[clone] (MeetP) {}
       +(0,1)       node[clone] (Meet) {}
      ++(-1,0.5)    node[clone] (MeetP0) {}

      ++(0.4,1.25)  node[clone] (MPT0inf) {}
       +(-0.5,1)    node[clone] (PT0inf) {}
      ++(-1,0.5)    node[clone] (MT0inf) {}
       +(-0.5,1)    node[clone] (T0inf) {}
        (MPT0inf)
      ++(1.5/4,1.5) node[clone] (MPT04) {}
       +(-0.5,1)    node[clone] (PT04) {}
      ++(-1,0.5)    node[clone] (MT04) {}
       +(-0.5,1)    node[clone] (T04) {}
       (MPT04)
      ++(0.25,1)    node[clone] (MPT03) {}
       +(-0.5,1)    node[clone] (PT03) {}
      ++(-1,0.5)    node[clone] (MT03) {}
       +(-0.5,1)    node[clone] (T03) {}
       (MPT03)
      ++(0.25,1)    node[clone] (MPT02) {}
       +(-0.5,1)    node[clone] (PT02) {}
      ++(-1,0.5)    node[clone] (MT02) {}
       +(-0.5,1)    node[clone] (T02) {}; % }}}
  % right wing nodes {{{
  \path (UP1)
      ++(0,1.5)      node[clone] (JoinP0) {}
      ++(1,-0.5)     node[clone] (JoinP) {}
       +(0,1)        node[clone] (Join) {}
      ++(1,0.5)      node[clone] (JoinP1) {}

      ++(-0.4,1.25)  node[clone] (MPT1inf) {}
       +(0.5,1)      node[clone] (PT1inf) {}
      ++(1,0.5)      node[clone] (MT1inf) {}
       +(0.5,1)      node[clone] (T1inf) {}
        (MPT1inf)
      ++(-1.5/4,1.5) node[clone] (MPT14) {}
       +(0.5,1)      node[clone] (PT14) {}
      ++(1,0.5)      node[clone] (MT14) {}
       +(0.5,1)      node[clone] (T14) {}
       (MPT14)
      ++(-0.25,1)    node[clone] (MPT13) {}
       +(0.5,1)      node[clone] (PT13) {}
      ++(1,0.5)      node[clone] (MT13) {}
       +(0.5,1)      node[clone] (T13) {}
       (MPT13)
      ++(-0.25,1)    node[clone] (MPT12) {}
       +(0.5,1)      node[clone] (PT12) {}
      ++(1,0.5)      node[clone] (MT12) {}
       +(0.5,1)      node[clone] (T12) {};  % }}}

  % central section edges  {{{
  \draw (F) -- (UP0) -- (UM) -- (UP1) -- (F)
    (F) -- (MeetP)
    (F) -- (JoinP)
    (F) to[curveR] (UD)
    (F) to[curveR] (AP)
    (F) to[out=90+20,in=-90-20] (DM)
    (UP0) -- (MeetP0)
    (UP0) -- (JoinP0)
    (UP0) to[out=90-15,in=-90] (AP0)
    (UP1) -- (JoinP1)
    (UP1) -- (MeetP1)
    (UP1) to[out=90+15,in=-90] (AP1)
    (UM) to[curveL] (U)
    (UM) -- (Meet)
    (UM) -- (Join)
    (UD) -- (U)
    (UD) to[curveR] (AD)
    (U) to[curveL] (A)
    (AP) -- (AP0) -- (A) -- (AP1) -- (AP) -- (AD) -- (A)
    (AP) to[curveL] (DP)
    (AP0) to[out=90,in=-90-15] (P0)
    (AP1) to[out=90,in=-90+15] (P1)
    (AD) to[curveR] (D)
    (A) to[curveL] (T)
    (DM) -- (MPT02) -- (MP) -- (MPT12) -- (DM) -- (DP) -- (D)
    (DP) to[curveL] (P)
    (D) to[curveR] (T)
    (MP) -- (MP0) -- (M) -- (MP1) -- (MP)
    (MP) to[curveR] (P)
    (MP0) -- (P0)
    (MP1) -- (P1)
    (M) to[curveL] (T)
    (P) -- (P0) -- (T) -- (P1) -- (P); % }}}
  % left wing edges {{{
  \draw (MeetP) -- (MeetP0) -- (Meet) -- (MeetP1) -- (MeetP)
    (MeetP) -- (MPT0inf)
    (MeetP1) to[out=90-10,in=-90-20] (MP1)
    (MeetP0) -- (MT0inf)
    (Meet) -- (M)
    (MPT0inf) -- (MT0inf) -- (T0inf) -- (PT0inf) -- (MPT0inf)
    (MPT0inf) edge[dotted] (MPT04)
    (MT0inf) edge[dotted] (MT04)
    (PT0inf) edge[dotted] (PT04)
    (T0inf) edge[dotted] (T04)
    (MPT04) -- (MT04) -- (T04) -- (PT04) -- (MPT04)
    (MPT04) -- (MPT03) (MT04) -- (MT03) (PT04) -- (PT03) (T04) -- (T03)
    (MPT03) -- (MT03) -- (T03) -- (PT03) -- (MPT03)
    (MPT03) -- (MPT02) (MT03) -- (MT02) (PT03) -- (PT02) (T03) -- (T02)
    (MPT02) -- (MT02) -- (T02) -- (PT02) -- (MPT02)
    (MT02) -- (MP0)
    (T02) -- (P0)
    (PT02) -- (P);  % }}}
  % right wing edges {{{
  \draw (JoinP) -- (JoinP0) -- (Join) -- (JoinP1) -- (JoinP)
    (JoinP) -- (MPT1inf)
    (JoinP0) to[out=90+10,in=-90+20] (MP0)
    (JoinP1) -- (MT1inf)
    (Join) -- (M)
    (MPT1inf) -- (MT1inf) -- (T1inf) -- (PT1inf) -- (MPT1inf)
    (MPT1inf) edge[dotted] (MPT14)
    (MT1inf) edge[dotted] (MT14)
    (PT1inf) edge[dotted] (PT14)
    (T1inf) edge[dotted] (T14)
    (MPT14) -- (MT14) -- (T14) -- (PT14) -- (MPT14)
    (MPT14) -- (MPT13) (MT14) -- (MT13) (PT14) -- (PT13) (T14) -- (T13)
    (MPT13) -- (MT13) -- (T13) -- (PT13) -- (MPT13)
    (MPT13) -- (MPT12) (MT13) -- (MT12) (PT13) -- (PT12) (T13) -- (T12)
    (MPT12) -- (MT12) -- (T12) -- (PT12) -- (MPT12)
    (MT12) -- (MP1)
    (T12) -- (P1)
    (PT12) -- (P);  % }}}

  % wing clone labels
  \node[cloneLabel, above, xshift=-0.25em] at (Meet) {$\CloneWedge$};
  \node[cloneLabel, above, xshift=0.25em] at (Join) {$\CloneVee$};
  \node[cloneLabel, below left, xshift=0.43em] at (MPT0inf) {$\Clone{MPT}_0^\infty$};
  \node[cloneLabel, below right, xshift=-0.43em] at (MPT1inf) {$\Clone{MPT}_1^\infty$};
  \node[cloneLabel, left, yshift=0.2em] at (T) {$\CloneTop$};

  % central clone labels
  \node[cloneLabel, left, yshift=-0.25em] at (F) {$\CloneBot$};
  \node[cloneLabel, left] at (U) {$\Clone{U}$};
  \node[cloneLabel, above, xshift=0.3em] at (A) {$\Clone{A}$};
  \node[cloneLabel, left, align=left] at (AP0) {$\Clone{AP}_0$ \\ $\approx \ZZ_2$};
  \node[cloneLabel, below] at (AP) {$\Clone{AP}$};
  \node[cloneLabel, below right, xshift=-0.4em] at (DM) {$\Clone{DM}$};

  \node (center) at ($(current bounding box.north east)!0.5!(current bounding box.south west)$) {\phantom{$\cdot$}};
\end{tikzpicture} % }}}3
  \hspace{4em}
  \renewcommand{\tabcolsep}{1em}
  \begin{tabular}{cm{5.1em}}
    Clone                  & Generating \newline Operations \\ \midrule
    $\CloneBot$            & $\emptyset$            \\
    $\CloneWedge$          & $x \wedge y$, $0$, $1$ \\
    $\CloneVee$            & $x \vee y$, $0$, $1$   \\
    $\Clone{U}$            & $\neg x$, $0$          \\
    $\Clone{MPT}_0^\infty$ & $x \vee (y \wedge z)$  \\
    $\Clone{MPT}_1^\infty$ & $x \wedge (y \vee z)$  \\
    $\Clone{AP}$           & $x+y+z$                \\
    $\Clone{AP}_0$         & $x+y$                  \\
    $\Clone{A}$            & $x\Aiff y$, $0$        \\
    $\Clone{DM}$           & $\text{maj}(x,y,z)$    \\
    $\CloneTop$            & $x \vee y$, $\neg x$
  \end{tabular}
  \caption{ Diagram of Post's Lattice, the lattice of all clones on a 2-element
    domain (i.e.\ boolean clones). Clones which will be relevant later have been
    labelled, and operations generating each of them are given in the table
    to the right. Simon's Problem corresponds to the $\HKP$ on the operational
    clone of the group $\ZZ_2$, namely $\Clone{AP}_0$. }
  \Description{See caption.}
  \label{fig:posts_lattice}
\end{figure*}  % }}}2

As described in the introduction, given a group $\alg{G}$, the function $f: G
\to X$ \emph{hides} a subgroup $\alg{D}\leq \alg{G}$ if $f(g) = f(h)$ if and
only if $gD = hD$.
\begin{center} \begin{tabular}{rl}
  \multicolumn{2}{l}{\textbf{Hidden Subgroup Problem (HSP)}} \\ \midrule
  \texttt{Input:}  & group $\alg{G}$, \\
    & oracle $f: G \to X$ hiding some subgroup  \\[0.25em]
  \texttt{Task:} & determine the subgroup $\alg{D}\leq \alg{G}$ that $f$ hides
\end{tabular} \end{center}
\noindent
A quantum polynomial-time solution to the HSP for \emph{abelian} $\alg{G}$ is
known~\cite{Kitaev_AbStab}, and this case contains many famous quantum
algorithms exhibiting super-polynomial speedup over classical solutions. The HSP
for \emph{non-abelian} $\alg{G}$ has been solved in some quite specific cases
and restrictive settings, but it remains open in general. In spite of focused
research, no general polynomial-time quantum algorithm for the HSP has been
found.

That the HSP has been efficiently solved for a large class of groups (abelian
ones) and an irregular constellation of groups otherwise suggests two
possibilities.
\begin{enumerate}[leftmargin=*]
  \item There is no quantum polynomial-time solution to the general HSP. It is
    suspected that quantum algorithms cannot solve \comp{NP-complete} problems
    in polynomial time, and it may be the case that the general HSP is
    \comp{NP-complete}.
  \item There is a quantum polynomial-time solution, but it requires a deeper
    insight into the problem than has been obtained over the past 30 years.
\end{enumerate}
In either of these cases, large additional classes of structures for which the
HSP has an efficient solution would be incredibly valuable. To this end, we
propose a generalization of the HSP to universal algebraic structures. Such a
generalization will extend the domain of the problem to classes of structures
with additional gradations of complexity, with the expectation that these these
gradations will be reflected in the complexity of the generalized problem.

\subsection{The Hidden Kernel Problem} \label{subsec:HKP} % {{{2
The Hidden Subgroup Problem can be framed as a problem about the \emph{quotient
structure} induced by the hidden subgroup $\alg{D}$, namely $G/D$ (the set of
cosets of $\alg{D}$ in $\alg{G}$). In this formulation, we are given $\alg{G}$
and a \emph{bijective} function $f: G/D\to X$ as an oracle, and tasked with
computing $\alg{D}\leq \alg{G}$. The set $G/D$ does \emph{not} have group
structure itself unless $\alg{D}$ is a normal subgroup. The progress which has
been made on the general HSP begins with a deep analysis of what limited
structure $G/D$ \emph{does} have for certain classes of groups.

For abelian groups, every subgroup is normal, so the structure of $G/D$ is quite
uniform. Restricting attention to only the \emph{normal} subgroups yields the
\emph{Hidden Normal Subgroup Problem}. It is this problem which we generalize to
obtain the \emph{Hidden Kernel Problem}.

Let $\AA$ be an algebra and $\theta$ a congruence of $\AA$. Following the
definition of a hidden subgroup, the function $f: A\to X$ \emph{hides} the
congruence $\theta$ if $f(a) = f(b)$ precisely when $a \;\theta\; b$.
\begin{center} \begin{tabular}{rl}
  \multicolumn{2}{l}{\textbf{The Hidden Kernel Problem (preliminary version)}} \\ \midrule
  \texttt{Input:}  & algebra $\alg{A}$, \\
    & oracle $f : A \to X$ hiding some congruence \\[0.25em]
  \texttt{Task:} & determine the congruence $\theta$ of $\alg{A}$ that $f$ hides
\end{tabular} \end{center}
It is not hard to show that a function $f: \alg{A}\to X$ hides a congruence
$\theta$ if and only if $f$ is a homomorphism with $\theta = \ker(f)$. This
observation leads to the general statement of the Hidden Kernel Problem.
\begin{center} \begin{tabular}{rl}
  \multicolumn{2}{l}{\textbf{The Hidden Kernel Problem (HKP)}} \\ \midrule
  \texttt{Input:}  & similar algebras $\alg{A}$ and $\alg{B}$, \\
    & homomorphism $\phi : \alg{A}\to \alg{B}$ (as an oracle) \\[0.25em]
  \texttt{Task:} & determine the congruence $\ker(\phi)$
\end{tabular} \end{center}
This generalization includes and extends the Hidden Normal Subgroup Problem, and
abstracts away the somewhat unnatural notion of ``hiding''.
%---------------------------------------------------------------------------}}}2

\subsection{Simon's Problem and Post's Lattice} \label{subsec:simon_post} % {{{2
Perhaps the simplest instance of the Hidden Subgroup Problem is when the group
is an elementary abelian 2-group, $\alg{G} = \ZZ_2^n$. This instance of the
$\HSP$ is known as \emph{Simon's Problem}~\cite{Simon_Alg}, and historically is
one of the first examples of super-polynomial speedup when using a quantum
algorithm.

We will consider a universal algebraic version of Simon's Problem --- The Hidden
Kernel Problem on $\AA = \alg{B}^n$, where $\alg{B}$ is a 2-element algebra.
There are infinitely many such $\alg{B}$, and \emph{a priori} there seems to be
neither a natural starting point nor a systematic way in which to proceed. One
crucial observation is the following.

\begin{obs} \label{obs:clones_and_hkp} % {{{
Let $\alg{B}$ and $\alg{D}$ be universal algebras such that $B = D$ and
$\Clo(\alg{B}) \subseteq \Clo(\alg{D})$. Each specific instance of
$\HKP(\alg{D}^n)$ is a specific instance of $\HKP(\alg{B}^n)$. A solution to all
instances of $\HKP(\alg{B}^n)$ yields a solution to all instances of
$\HKP(\alg{D}^n)$.
\end{obs} % }}}

This observation follows from the fact that every homomorphism of $\alg{D}^n$ is
also a homomorphism of $\alg{B}^n$ when $\Clo(\alg{B}) \subseteq \Clo(\alg{D})$
($\Leftrightarrow \alg{B}\preceq \alg{D}$). The observation implies that
$\HKP(\alg{B}^n)$ is equivalent to $\HKP(\alg{D}^n)$ whenever $\Clo(\alg{B}) =
\Clo(\alg{D})$. Instead of considering all possible 2-element algebras
$\alg{B}$, we therefore consider only those algebras with distinct clones of
term operations.

The task of describing all clones of operations on a 2-element domain (also
known as boolean clones), was famously undertaken in 1941 by Emil
Post~\cite{Post_Lattice}. Ordered by inclusion, the set of all boolean clones is
known as \emph{Post's Lattice}, and has a particularly regular structure (see
Figure~\ref{fig:posts_lattice}). The operations given in
Figure~\ref{fig:posts_lattice} are in terms of the familiar boolean operations,
with the possible exception of the \emph{majority} function, defined
\[
  \maj(x,y,z)
  \coloneqq \big( x \wedge y \big) \vee (x \wedge z) \vee (y \wedge z).
\]
An important property of this operation is that 
\[
  x 
  = \maj(y,x,x)
  = \maj(x,y,x)
  = \maj(x,x,y).
\]

We end this subsection by defining the quantum circuit which will yield a
polynomial-time quantum solution to $\HKP(\alg{B}^n)$ for many of the algebras
we consider. This circuit is essentially the same as Simon's original
circuit~\cite{Simon_Alg}.

\begin{defn} \label{defn:simon_circuit} % {{{
Let $\phi: \{0,1\}^n \to \{0,1\}^m$, and define the $2^{n+m}$-dimensional
unitary transformation $\widehat{\phi}$ by its action on basis vectors,
\[
  \widehat{\phi} \big( \ket{x} \otimes \ket{y} \big)
  \coloneqq \ket{x} \otimes \ket{y + \phi(x)}, 
\]
where $x\in \{0,1\}^n$, $y\in \{0,1\}^m$, and ``+'' is componentwise addition
modulo $2$. The quantum circuit $\Simon_\phi(n)$ is the circuit
\[
  \Simon_\phi(n)
  \coloneqq \big( H^{\otimes n} \otimes I_m \big) \;\widehat{\phi}\;
    \big( H^{\otimes n} \otimes I_m \big)
\]
of size $n+m$, where measurement is performed on the first $n$ qubits. The
circuit $\Simon_\phi(n)$ is given graphically in Figure~\ref{fig:simon}. Note
that we can always take $m\leq n$, so the circuit is of size $\Theta(n)$.
\end{defn}  % }}}

% fig:simon {{{3
\begin{figure} \centering 
\[
  \Simon_\phi(n) 
  \quad = \quad
% circuit {{{
\begin{tikzpicture}[font=\scriptsize, baseline=(center.base)]
  \path (0,0) node (n1) {$\ket{0}$}
      ++(1,0) node[gate-c] (n1H1) {$H$}
      ++(1,0) coordinate (n1phi)
      ++(1,0) node[gate-c] (n1H2) {$H$}
      ++(1,0) node[gate-r] (n1meas) {\tikzMeasure}
      ++(0.75,0) node (n1end) {}
        (n1)
      ++(0,-0.5) node (ndots1) {$\vdots$}
      ++(0,-0.75) node (nn) {$\ket{0}$}
      ++(0,-0.5) node (m1) {$\ket{0}$}
      ++(0,-0.5) node (mdots1) {$\vdots$}
      ++(0,-0.75) node (mm) {$\ket{0}$};
  \coordinate (t) at ($(n1)!0.5!(mm)$);
  \path (n1H1|-nn) node[gate-c] (nnH1) {$H$}
        (n1phi|-t) node[gate-r, minimum height=10em, minimum width=1.55em] (phi) {$\widehat{\phi}$}
        (n1H2|-nn) node[gate-c] (nnH2) {$H$}
        (n1meas|-nn) node[gate-r] (nnmeas) {\tikzMeasure}
        (n1end|-nn) node (nnend) {}
        (n1end|-m1) node (m1end) {}
        (n1end|-mm) node (mmend) {}
        (n1meas|-ndots1) node (ndots2) {$\vdots$}
        (n1meas|-mdots1) node (mdots2) {$\vdots$};
  \path[->] (n1) edge (n1H1) (n1H1) edge (n1-|phi.west) (phi.east|-n1) edge (n1H2) (n1H2) edge (n1meas) (n1meas) edge (n1end)
            (nn) edge (nnH1) (nnH1) edge (nn-|phi.west) (phi.east|-nn) edge (nnH2) (nnH2) edge (nnmeas) (nnmeas) edge (nnend)
            (m1) edge (m1-|phi.west) (phi.east|-m1) edge (m1end)
            (mm) edge (mm-|phi.west) (phi.east|-mm) edge (mmend);
  \draw[densely dotted] (n1end.north west) -- (n1end.north east-|n1end.east) -- node[auto]{$n$} (nnend.south east-|nnend.east) -- (nnend.south west);
  \draw[densely dotted] (m1end.north west) -- (m1end.north east-|m1end.east) -- node[auto]{$m$} (mmend.south east-|mmend.east) -- (mmend.south west);
  \node (center) at ($(current bounding box.north east)!0.5!(current bounding box.south west)$) {\phantom{$\cdot$}};
\end{tikzpicture} % }}}
\]
  \caption{ The circuit $\Simon_\phi(n)$ of Definition~\ref{defn:simon_circuit}. }
  \Description{See Definition~\ref{defn:simon_circuit}.}
  \label{fig:simon}
\end{figure}
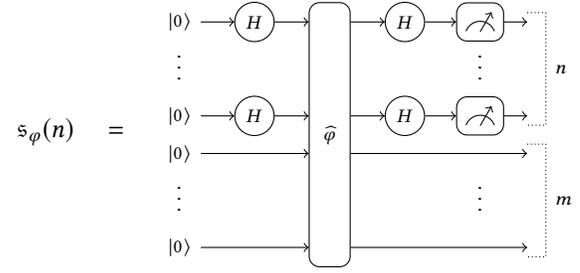  % }}}3
%----------------------------------------------------------------------------}}}2
%----------------------------------------------------------------------------}}}1

\section{The HKP and Post's Lattice: Algorithms} \label{sec:hkp_positive}  % {{{1
In this section we present positive results for solving the Hidden Kernel
Problem for Post's Lattice in polynomial time using classical and
quantum algorithms. The main quantum result is Theorem~\ref{thm:AP_MPT_quantum}
and the main classical result is Corollary~\ref{cor:classical_CD}, stated below.

\begin{thm} \label{thm:AP_MPT_quantum} % {{{
Let $\alg{B}$ be a 2-element algebra such that one of
\begin{enumerate}
  \item $\Clone{MPT}_0^{\infty} \preceq \alg{B}$,
  \item $\Clone{MPT}_1^{\infty} \preceq \alg{B}$, or
  \item $\Clone{AP} \preceq \alg{B}$
\end{enumerate}
holds and let $\AA = \alg{B}^n$. The quantum circuit $\Simon_\phi(n)$ of
Definition~\ref{defn:simon_circuit} solves $\HKP(\AA)$ with probability $1 -
1/\tau$ in $\Theta(n + \lg(\tau))$ iterations.
\end{thm} % }}}

\begin{cor} \label{cor:classical_CD}  % {{{
Let $\alg{B}$ be a 2-element algebra such that one of
\begin{enumerate}
  \item $\Clone{MPT}_0^{\infty} \preceq \alg{B}$,
  \item $\Clone{MPT}_1^{\infty} \preceq \alg{B}$, or
  \item $\Clone{DM} \preceq \alg{B}$
\end{enumerate}
holds and let $\AA = \alg{B}^n$. Then $\HKP(\AA)$ is solvable classically in
$\O(n)$ time.
\end{cor} % }}}

The best way to summarize these two results is in terms of Post's Lattice, see
Figure~\ref{fig:PL_positive}. In particular, powers of algebras whose clones are
contained in the interval between $\Clone{AP}$ and $\Clone{A}$ have $\HKP$
solvable in polynomial-time using a quantum algorithm, but not a classical one
(see Theorem~\ref{thm:no_classical}). Combined with the results of the next
section, this provides a complete quantum and classical classification of the
algorithmic complexity of the $\HKP$ for powers of 2-element algebras.

% fig:PL_positive {{{2
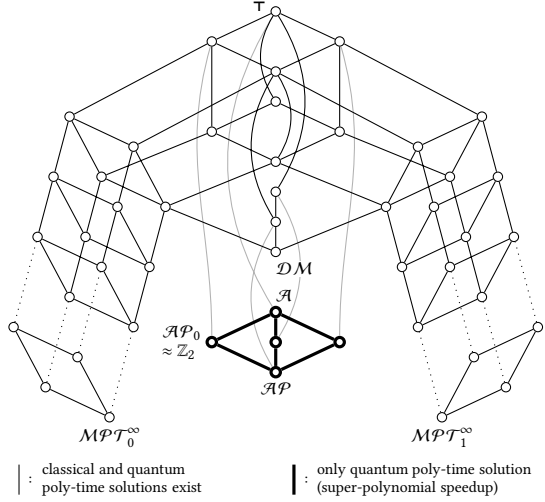
\begin{figure} \centering 
% lattice PL, positive {{{3
\begin{tikzpicture}[font=\tiny, yscale=0.8, xscale=0.85]
  \node at ($(-2.6, 3.25)-(0,4)$) (MPT0infX) {};
  \node at ($(2.6, 3.25)-(0,4)$)  (MPT1infX) {};
  % central section nodes {{{
  \path (0, 0)    node[very thick, clone] (AP) {}
        +(-1,0.5) node[very thick, clone] (AP0) {}
        +(0,0.5)  node[very thick, clone] (AD) {}
        +(1,0.5)  node[very thick, clone] (AP1) {}
       ++(0,1)    node[very thick, clone] (A) {}
       ++(0,1)    node[clone] (DM) {}
       ++(0,0.5)  node[clone] (DP) {}
       ++(0,0.5)  node[clone] (D) {}
       ++(0,0.5)  node[clone] (MP) {}
        +(-1,0.5) node[clone] (MP0) {}
        +(1,0.5)  node[clone] (MP1) {}
       ++(0,1)    node[clone] (M) {}
       ++(0,0.5)  node[clone] (P) {}
        +(-1,0.5) node[clone] (P0) {}
        +(1,0.5)  node[clone] (P1) {}
       ++(0,1)    node[clone] (T) {}; % }}}
  % left wing nodes {{{
  \path (MPT0infX)  node[clone] (MPT0inf) {}
       +(-0.5,1)    node[clone] (PT0inf) {}
      ++(-1,0.5)    node[clone] (MT0inf) {}
       +(-0.5,1)    node[clone] (T0inf) {}
        (MPT0inf)
      ++(1.5/4,1.5) node[clone] (MPT04) {}
       +(-0.5,1)    node[clone] (PT04) {}
      ++(-1,0.5)    node[clone] (MT04) {}
       +(-0.5,1)    node[clone] (T04) {}
       (MPT04)
      ++(0.25,1)    node[clone] (MPT03) {}
       +(-0.5,1)    node[clone] (PT03) {}
      ++(-1,0.5)    node[clone] (MT03) {}
       +(-0.5,1)    node[clone] (T03) {}
       (MPT03)
      ++(0.25,1)    node[clone] (MPT02) {}
       +(-0.5,1)    node[clone] (PT02) {}
      ++(-1,0.5)    node[clone] (MT02) {}
       +(-0.5,1)    node[clone] (T02) {}; % }}}
  % right wing nodes {{{
  \path (MPT1infX)   node[clone] (MPT1inf) {}
       +(0.5,1)      node[clone] (PT1inf) {}
      ++(1,0.5)      node[clone] (MT1inf) {}
       +(0.5,1)      node[clone] (T1inf) {}
        (MPT1inf)
      ++(-1.5/4,1.5) node[clone] (MPT14) {}
       +(0.5,1)      node[clone] (PT14) {}
      ++(1,0.5)      node[clone] (MT14) {}
       +(0.5,1)      node[clone] (T14) {}
       (MPT14)
      ++(-0.25,1)    node[clone] (MPT13) {}
       +(0.5,1)      node[clone] (PT13) {}
      ++(1,0.5)      node[clone] (MT13) {}
       +(0.5,1)      node[clone] (T13) {}
       (MPT13)
      ++(-0.25,1)    node[clone] (MPT12) {}
       +(0.5,1)      node[clone] (PT12) {}
      ++(1,0.5)      node[clone] (MT12) {}
       +(0.5,1)      node[clone] (T12) {};  % }}}

  % central section edges  {{{
  \draw[misc] (AP) to[curveL] (DP)
    (AP0) to[out=90,in=-90-15] (P0)
    (AP1) to[out=90,in=-90+15] (P1)
    (AD) to[curveR] (D)
    (A) to[curveL] (T);
  \draw[superpoly] (AP) -- (AP0) -- (A) -- (AP1) -- (AP) -- (AD) -- (A);
  \draw[class_and_quant] (DM) -- (MPT02) -- (MP) -- (MPT12) -- (DM) -- (DP) -- (D)
    (DP) to[curveL] (P)
    (D) to[curveR] (T)
    (MP) -- (MP0) -- (M) -- (MP1) -- (MP)
    (MP) to[curveR] (P)
    (MP0) -- (P0)
    (MP1) -- (P1)
    (M) to[curveL] (T)
    (P) -- (P0) -- (T) -- (P1) -- (P); % }}}
  % left wing edges {{{
  \draw[class_and_quant] (MPT0inf) -- (MT0inf) -- (T0inf) -- (PT0inf) -- (MPT0inf)
    (MPT0inf) edge[dotted] (MPT04)
    (MT0inf) edge[dotted] (MT04)
    (PT0inf) edge[dotted] (PT04)
    (T0inf) edge[dotted] (T04)
    (MPT04) -- (MT04) -- (T04) -- (PT04) -- (MPT04)
    (MPT04) -- (MPT03) (MT04) -- (MT03) (PT04) -- (PT03) (T04) -- (T03)
    (MPT03) -- (MT03) -- (T03) -- (PT03) -- (MPT03)
    (MPT03) -- (MPT02) (MT03) -- (MT02) (PT03) -- (PT02) (T03) -- (T02)
    (MPT02) -- (MT02) -- (T02) -- (PT02) -- (MPT02)
    (MT02) -- (MP0)
    (T02) -- (P0)
    (PT02) -- (P);  % }}}
  % right wing edges {{{
  \draw[class_and_quant] (MPT1inf) -- (MT1inf) -- (T1inf) -- (PT1inf) -- (MPT1inf)
    (MPT1inf) edge[dotted] (MPT14)
    (MT1inf) edge[dotted] (MT14)
    (PT1inf) edge[dotted] (PT14)
    (T1inf) edge[dotted] (T14)
    (MPT14) -- (MT14) -- (T14) -- (PT14) -- (MPT14)
    (MPT14) -- (MPT13) (MT14) -- (MT13) (PT14) -- (PT13) (T14) -- (T13)
    (MPT13) -- (MT13) -- (T13) -- (PT13) -- (MPT13)
    (MPT13) -- (MPT12) (MT13) -- (MT12) (PT13) -- (PT12) (T13) -- (T12)
    (MPT12) -- (MT12) -- (T12) -- (PT12) -- (MPT12)
    (MT12) -- (MP1)
    (T12) -- (P1)
    (PT12) -- (P);  % }}}

  % wing clone labels
  \node[cloneLabel, below] at (MPT0inf) {$\Clone{MPT}_0^\infty$};
  \node[cloneLabel, below] at (MPT1inf) {$\Clone{MPT}_1^\infty$};
  \node[cloneLabel, left, yshift=0.2em] at (T) {$\CloneTop$};

  % central clone labels
  \node[cloneLabel, above, xshift=0.3em] at (A) {$\Clone{A}$};
  \node[cloneLabel, left, align=left] at (AP0) {$\Clone{AP}_0$ \\ $\approx \ZZ_2$};
  \node[cloneLabel, below] at (AP) {$\Clone{AP}$};
  \node[cloneLabel, below right, xshift=-0.4em] at (DM) {$\Clone{DM}$};

  % key
  \coordinate (key_center) at ($(current bounding box.north east)!0.5!(current bounding box.south west)$);
  \coordinate (class_and_quant_key) at ($(key_center)+(-3.8,-4.25)$);
  \draw[class_and_quant] ($(class_and_quant_key.west)+(-0.215,-0.25)$) -- node[black, right]{:} ($(class_and_quant_key.west)+(-0.215,0.25)$);
  \node[anchor=west, align=left] at (class_and_quant_key) {classical and quantum \\ poly-time solutions exist};
  \coordinate (superpoly_key) at ($(key_center)+(0.5,-4.25)$);
  \draw[superpoly] ($(superpoly_key.west)+(-0.215,-0.25)$) -- node[black, right]{:} ($(superpoly_key.west)+(-0.215,0.25)$);
  \node[anchor=west, align=left] at (superpoly_key) {only quantum poly-time solution \\ (super-polynomial speedup)};
\end{tikzpicture} % }}}3
  \caption{ The fragment of Post's Lattice for which the $\HKP$ is solvable in
    polynomial-time using a quantum algorithm (whole diagram). Bold lines
    indicate clones exhibiting a super-polynomial speedup with a quantum
    algorithm (i.e.\ no classical polynomial-time algorithm exists, see
    Theorem~\ref{thm:no_classical}), while plain lines indicate clones for which
    a classical polynomial-time algorithm exists. }
  \Description{ The fragment of Post's Lattice consisting of clones which
    contain $\Clone{MPT}_0^\infty$, $\Clone{MPT}_1^\infty$, or $\Clone{AP}$. The
    interval $[\Clone{AP}, \Clone{A}]$ exhibits superpolynomial-speedup with a
    quantum algorithm. }
  \label{fig:PL_positive}
\end{figure}  % }}}2

The remainder of the section is devoted to proving
Theorem~\ref{thm:AP_MPT_quantum} and Corollary~\ref{cor:classical_CD}. We
proceed by proving a series of lemmas, beginning first with two lemmas leading
up to the famous result of Simon~\cite{Simon_Alg}, which we phrase and prove in
terms of the Hidden Kernel Problem.

\begin{lem} \label{lem:simon_dual}  % {{{
Let $\alg{D} \leq \big( \ZZ_2^n \big)^2$. For all $g\in \big( \ZZ_2^n \big)^2$,
\[
  \sum_{d\in D} (-1)^{g\cdot d}
  = \begin{cases}
    |D| & \text{if } g \cdot d = 0 \text{ for all } d\in D, \\
    0   & \text{otherwise},
  \end{cases}
\]
where $g\cdot d$ is the dot product modulo 2 of $g$ and $d$ (regarded as
$\ZZ_2$-vectors).
\end{lem}
\begin{proof}
Define sets
\[
  D_0
  = \big\{ d\in D \mid g\cdot d = 0 \big\}
  \quad\text{and}\quad
  D_1
  = \big\{ d\in D \mid g\cdot d = 1 \big\}.
\]
Observe that these sets are disjoint and $D = D_0\cup D_1$, so they partition
$D$. Thus,
\[
  \sum_{d\in D} (-1)^{g\cdot d}
  = |D_0| - |D_1|.
\]
It is not hard to see that $D_0$ is a subgroup of $\alg{D}$. If $D_1 =
\emptyset$, then the summation is equal to $|D|$ and $g\cdot d = 0$ for all $d\in
D$. Let us assume therefore that $D_1 \neq \emptyset$ and let $d\in D_1$. It
follows that $d + D_0 \subseteq D_1$ and $d + D_1 \subseteq D_0$. Hence $D_1 = d
+ D_0$ and so $D_1$ is a coset of $D_0$. Therefore $|D_1| = |D_0|$ and the
summation equals $0$.
\end{proof} % }}}

\begin{lem} \label{lem:simon_generates} % {{{
Let $\theta$ be a congruence of $\ZZ_2^n$ and let $X\subseteq \theta$. The set
$X$ generates $\theta$ with probability $1 - 1/\tau$, where $|X| \in \Theta(n +
\lg(\tau))$. Specifically, if $|X| = 2n$ then the probability of generation is
$1 - 2^{-n}$.
\end{lem}
\begin{proof}
Let $\psi$ be an arbitrary congruence of $\ZZ_2^n$. We have $g \;\psi\; h$ if
and only if $(g+h) \;\psi\; 0$, so the congruence class of $\psi$ containing
$0\in \ZZ_2^n$ uniquely characterizes $\psi$. Call this congruence class
$H_\psi$,
\[
  H_\psi
  = \big\{ g\in \ZZ_2^n \mid g \;\psi\; 0 \big\}.
\]
It's not hard to see that $H_\psi$ is a subgroup of $\ZZ_2^n$. In this manner
subgroups (normal subgroups, in general) correspond to congruences.

Let $X_0 = \{ x+y \mid (x,y)\in X \}$ and note that $X_0$ is contained in the
congruence class of $0$ in $\theta$. From the paragraph above, we have $\Cg(X)
\neq \theta$ if and only if $\Sg(X_0) \neq \alg{H}_\theta$. The set $X_0$ fails
to generate $\alg{H}_\theta$ if and only if $X_0$ is contained in some maximal
subgroup $\alg{M}$ of $\alg{H}_\theta$. Since $\alg{H}_\theta$ is abelian and a
subgroup of $\ZZ_2^n$, these are all of size $|H_\theta|/2$. Thus,
\begin{align*}
  &\Prob\big( \Cg(X) \neq \theta \big)
  = \Prob\big( \Sg(X_0) \neq \alg{H}_\theta \big) \\
  &\qquad = \Prob\big( \exists \text{ maximal } \alg{M}\leq \alg{H}_\theta \text{ with } X_0\subseteq M  \big) \\
  &\qquad \leq \sum_{\text{maximal } \alg{M}\leq \alg{H}_\theta} \Prob\big( X_0\subseteq M \big) \\
  &\qquad = \sum_{\text{maximal } \alg{M}\leq \alg{H}_\theta} \; \prod_{g\in X_0} \Prob\big( g\in M \big) \\
  &\qquad = \sum_{\text{maximal } \alg{M}\leq \alg{H}_\theta} \; \prod_{g\in X_0} \frac{1}{2}
  = \frac{L}{2^{|X_0|}},
\end{align*}
where $L$ is the number of maximal subgroups of $\alg{H}_\theta$. We now
endeavor to calculate $L$.

Since $\ZZ_2^n$ is an elementary abelian $2$-group, $\alg{H}_\theta$ is as well.
It follows that $\alg{H}_\theta$ can be regarded as a $\ZZ_2$-vector space.
Maximal subgroups correspond to maximal subspaces, and each maximal subspace can
be uniquely characterized as the orthogonal complement of a $1$-dimensional
subspace. The number of $1$-dimensional subspaces is $|H_\theta| - 1$, and so is
at most $2^n - 1$. Combining everything, we have
\[
  \Prob\big( \Cg(X) \neq \theta \big)
  \leq \frac{L}{2^{|X_0|}}
  \leq \frac{2^n - 1}{2^{|X_0|}}.
\]
It follows that $X$ generates $\theta$ with probability $1 - 1/\tau$, where
$|X| \in \Theta(n + \lg(\tau))$, as claimed.
\end{proof} % }}}

\begin{thm}[Simon~\cite{Simon_Alg}] \label{thm:simon} % {{{
Let $\AA = \ZZ_2^n$. The quantum circuit $\Simon_\phi(n)$ of
Definition~\ref{defn:simon_circuit} solves $\HKP(\AA)$ with probability $1 -
1/\tau$ using $\Theta(n + \lg(\tau))$ iterations.
\end{thm}
\begin{proof}
Instead of calculating the pre-measurement output (call it $\ket{\psi}$), it
will be more convenient for us to calculate the density matrix of the
pre-measurement output, $\rho = \ket{\psi}\bra{\psi}$, and take the partial
trace.

The circuit operates on two registers, initially set to $\ket{\psi_0} = \ket{0^n}
\otimes \ket{0^m}$. This state has density matrix $\rho_0 =
\ket{\psi_0}\bra{\psi_0}$. After the first set of Hadamard gates, the density
matrix is
\begin{align*}
  \rho_1
  &\coloneqq (H^{\otimes n}\otimes I) \rho_0 (H^{\otimes n} \otimes I) \\
  &= \big( H^{\otimes n}\ket{0^n} \otimes \ket{0^m} \big)
    \big( \bra{0^n}H^{\otimes n} \otimes \bra{0^m} \big) \\
  &= \Big( \frac{1}{2^{n/2}} \sum_{x\in \{0,1\}^n} \ket{x} \otimes \ket{0^m} \Big)
    \Big( \frac{1}{2^{n/2}} \sum_{y\in \{0,1\}^n} \bra{y} \otimes \bra{0^m} \Big) \\
  &= \frac{1}{2^n} \sum_{x,y\in \{0,1\}^n} \ket{x}\bra{y} \otimes \ket{0^m}\bra{0^m}.
\end{align*}

We now describe the $\widehat{\phi}$ gate. A homomorphism $\phi: \AA\to \alg{B}$
is given as input as an oracle from which we may make queries. As a function,
$\phi$ can be regarded as mapping bit strings from $\{0,1\}^n$ to $\{0,1\}^m$.
Define the function $\widehat{\phi}$ by
\[ \begin{array}{rcl}
    \widehat{\phi} : \{0,1\}^{n+m} & \to     & \{0,1\}^{n+m}, \\
                            (a,b)  & \mapsto & (a, \phi(a)+b).
\end{array} \]
Observe that $\widehat{\phi}$ is invertible and that the original value of
$\phi(a)$ can be recovered --- $\widehat{\phi}(a,0) = (a,\phi(a))$. The function
$\widehat{\phi}$ can be regarded as a permutation of basis vectors of the space
$\B^{\otimes (n+m)}$, and thus can be extended to a unitary transformation
$\widehat{\phi}$ on the space.

Continuing with the evaluation of the circuit, after applying the
$\widehat{\phi}$ gate the density matrix is
\begin{align*}
  \rho_2
  &\coloneqq \widehat{\phi} \rho_1 \widehat{\phi}^{\dagger}
  = \frac{1}{2^n} \sum_{x,y\in \{0,1\}^n} 
    \big( \widehat{\phi}\ket{x} \otimes \ket{0^m} \big)
    \ \big( \widehat{\phi}\ket{y} \otimes \ket{0^m} \big)^{\dagger} \\
  & = \frac{1}{2^n} \sum_{x,y\in \{0,1\}^n} 
    \big( \ket{x} \otimes \ket{\phi(x)} \big)
    \ \big( \ket{y} \otimes \ket{\phi(y)} \big)^{\dagger} \\
  & = \frac{1}{2^n} \sum_{x,y\in \{0,1\}^n} 
    \ket{x}\bra{y} \otimes \ket{\phi(x)}\bra{\phi(y)}.
\end{align*}

Applying the last set of Hadamard gates yields the final pre-measurement density
matrix,
\begin{align*}
  \rho
  &\coloneqq \frac{1}{2^n} \sum_{x,y\in \{0,1\}^n} H^{\otimes n}\ket{x}\bra{y} H^{\otimes n} \otimes \ket{\phi(x)}\bra{\phi(y)} \\
  &= \frac{1}{2^{2n}} \sum_{\stack{a,b\in \{0,1\}^n,}{x,y\in \{0,1\}^n}}
    (-1)^{a\cdot x + b\cdot y} \ket{a}\bra{b} \otimes \ket{\phi(x)}\bra{\phi(y)}.
\end{align*}

The density matrix for the first register is given by taking the partial trace
of $\rho$ over the second register. Doing this, we have
\[
  \Tr_2(\rho)
  = \frac{1}{2^{2n}} \sum_{\stack{a,b\in \{0,1\}^n,}{x,y\in \{0,1\}^n}}
    (-1)^{a\cdot x + b\cdot y} \ket{a}\bra{b} \; \Tr\big( \ket{\phi(x)}\bra{\phi(y)} \big).
\]
It isn't hard to see that $\Tr(\ket{\phi(x)}\bra{\phi(y)}) = 1$ if $\phi(x) =
\phi(y)$ (i.e.\ if $(x,y)\in \ker(\phi)$) and is $0$ otherwise. Thus,
\[
  \Tr_2(\rho)
  = \frac{1}{2^{2n}} \sum_{a,b\in \{0,1\}^n} \ \sum_{(x,y)\in \ker(\phi)}
    (-1)^{a\cdot x + b\cdot y} \ket{a}\bra{b}.
\]

The congruence $\ker(\phi)$ is a subgroup of $(\ZZ_2^n)^2$ and we have $a\cdot x
+ b\cdot y = (a,b)\cdot (x,y)$. Lemma~\ref{lem:simon_dual} thus applies to the
inner summation. Applying it yields
\[
  \Tr_2(\rho)
  = \frac{|\ker(\phi)|}{2^{2n}} \sum_{(a,b)\in \ker(\phi)^\perp} \ket{a}\bra{b},
\]
where
\[
  \ker(\phi)^\perp 
  = \big\{ (a,b) \mid (a,b)\cdot (x,y) = 0 \text{ for all } (x,y)\in \ker(\phi) \big\}.
\]
The congruence $\ker(\phi)$ is a $\ZZ_2$-subspace of $\AA^2$, and this is its
orthogonal complement. It follows that $\ker(\phi)^\perp$ is also a congruence
of $\AA$ and $|\ker(\phi)||\ker(\phi)^\perp| = |A|^2 = 2^{2n}$. Thus
\[
  \Tr_2(\rho)
  = \frac{1}{|\ker(\phi)^\perp|} \sum_{(a,b)\in \ker(\phi)^\perp} \ket{a}\bra{b}.
\]
This corresponds to a uniform distribution over elements in $\ker(\phi)^\perp$.
By Lemma~\ref{lem:simon_generates}, iterating this procedure $\Theta(n +
\lg(\tau))$ times will produce a generating set for $\ker(\phi)^\perp$ with
probability $1 - 1/\tau$. 

As $\ker(\phi) = \big( \ker(\phi)^\perp \big)^\perp$, we have succeeded in
calculating $\ker(\phi)$ with probability $1 - 1/\tau$ using $\Theta(n +
\lg(\tau))$ iterations of a circuit of size $\Theta(n)$.
\end{proof} % }}}

At this point, we have established a quantum polynomial-time algorithm for the
$\HKP$ for powers of algebras above the clone $\Clone{AP}_0$ in Post's Lattice.
Aside from the group $\ZZ_2$, this includes only 3 other structures (see
Figure~\ref{fig:PL_positive}). The next lemma and theorem greatly expand this
list by deducing a classical algorithm for powers of certain congruence
distributive algebras which includes the infinite ``wings'' of Post's Lattice.

\begin{lem} \label{lem:classical_CD_projs}  % {{{
Let $\alg{B}_1, \dots, \alg{B}_n$ be similar algebras which are simple and let
$\AA = \prod_{i=1}^n \alg{B}_i$ be congruence distributive. Every congruence of
$\AA$ is the kernel of a projection homomorphism $\proj_I: \AA\to \prod_{i\in I}
\alg{B}_i$ for some $I\subseteq [n]$.
\end{lem}
\begin{proof}
For each $i\in [n]$ define $\eta_i = \ker(\proj_i)$. Let $\theta$ be a
congruence of $\AA$. Each algebra $\alg{B}_i$ is simple, so each congruence
$\eta_i$ is maximal in $\Con(\AA)$. It follows that
\[
  \theta \vee \eta_i = \begin{cases}
    \eta_i & \text{if } \theta\leq \eta_i, \\
    \Cong1 & \text{otherwise}.
  \end{cases}
\]
Using congruence distributivity, we now have
\[
  \theta
  = \theta \vee \Cong0
  = \theta \vee \bigwedge_{i\in [n]} \eta_i
  = \bigwedge_{i\in [n]} \big( \theta \vee \eta_i \big)
  = \bigwedge_{i\in I} \eta_i
  = \ker(\proj_I),
\]
where $I = \{ i \mid \theta\leq \eta_i \}$.
\end{proof} % }}}

\begin{thm} \label{thm:classical_CD_alg} % {{{
Let $\alg{B}_1, \dots, \alg{B}_n$ be similar algebras which are simple and let
$\AA = \prod_{i=1}^n \alg{B}_i$ be congruence distributive. There is a classical
$\cl{O}(n)$ algorithm solving $\HKP(\AA)$.
\end{thm}
\begin{proof}
Let $\phi:\AA\to \alg{D}$ be a homomorphism. Lemma~\ref{lem:classical_CD_projs}
implies that $\ker(\phi) = \ker(\proj_I)$ for some $I\subseteq [n]$. Thus,
specifying $I$ uniquely determines $\ker(\phi)$. Without loss of generality, we
may assume that $|B_i| \geq 2$ since we can eliminate any $\alg{B}_i$ of size
$1$ in the product and produce an algebra isomorphic to $\AA$. For each $i$, fix
distinct elements $a_i, b_i\in B_i$. We have that
\[
  i \not\in I
  \qquad \Leftrightarrow \qquad
  \phi(a_1, \dots, a_n) = \phi(a_1, \dots, b_i, \dots, a_n),
\]
where the inputs to $\phi$ differ only at position $i$. Starting at coordinate
$i=1$ and proceeding to coordinate $i=n$, we evaluate both sides of the
equality, recording when it holds and when it fails. The set of $i$ for which
the equality fails is $I$. This is an $\cl{O}(n)$ procedure.
\end{proof} % }}}

We are now ready to prove Corollary~\ref{cor:classical_CD}, which we began the
section with.

\begin{customcor}{\ref{cor:classical_CD}}[redux]  % {{{
Let $\alg{B}$ be a 2-element algebra such that one of
\begin{enumerate}
  \item $\Clone{MPT}_0^{\infty} \preceq \alg{B}$,
  \item $\Clone{MPT}_1^{\infty} \preceq \alg{B}$, or
  \item $\Clone{DM} \preceq \alg{B}$
\end{enumerate}
holds and let $\AA = \alg{B}^n$. Then $\HKP(\AA)$ is solvable classically in
$\O(n)$ time.
\end{customcor}
\begin{proof}
By Theorem~\ref{thm:classical_CD_alg}, it is sufficient to show that $\AA$ is
congruence distributive. The primary tool that we use is the existence of
J\'{o}nsson term operations~\cite{Jonsson_CDTerms}, which is equivalent to $\AA$
being congruence distributive. Briefly, these are a collection of 3-ary term
operations of $\AA$, $J_1$, $\dots$, $J_{2m+1}$, satisfying the universally
quantified identities
\[ \begin{aligned}
  & J_1(x,x,y) = x, \quad\qquad J_{2m+1}(x,y,y) = y, \\
  & \begin{aligned}
    & J_i(x,y,x) = x,                   & \quad & \text{for } i\in [2m+1], \\
    & J_{2i+1}(x,y,y) = J_{2i+2}(x,y,y) & \quad & \text{for } i\in [m-1], \\
    & J_{2i}(x,x,y) = J_{2i+1}(x,x,y)   & \quad & \text{for } i\in [m].
  \end{aligned}
\end{aligned} \]
We are now ready to proceed with the proof.

Suppose that $\Clone{DM} \preceq \alg{B}$ and let $\AA = \alg{B}^n$. The clone
$\Clone{DM}$ has $\maj(x,y,z)$ as a generating operation (see the discussion at
the end of Section~\ref{sec:HKP} and Figure~\ref{fig:posts_lattice}), and so
this is also an operation of $\alg{B}$ and hence of $\AA$ (where it is extended
from $\alg{B}$ componentwise). Simply taking $J_1(x,y,z) \coloneqq \maj(x,y,z)$
and verifying that this choice satisfies the J\'{o}nsson identities is
sufficient to prove that $\AA$ is congruence distributive.

The arguments for $\Clone{MPT}_0^\infty$ and $\Clone{MPT}_1^\infty$ are quite
similar, so we will present only the $\Clone{MPT}_0^\infty \preceq \alg{B}$
case. Suppose that $\Clone{MPT}_0^\infty \preceq \alg{B}$ and let $\AA =
\alg{B}^n$. Recall that $\Clone{MPT}_0^\infty$ has $x \wedge (y\vee z)$ as a
generating operation, and so this is an operation of $\alg{A}$. Observe that $x
\wedge y = x \wedge (y\vee y)$, so $\AA$ also has $x\wedge y$ as an operation.
Define term operations of $\AA$,
\begin{align*}
  J_1(x,y,z) &\coloneqq x \wedge (y \vee z),
    & J_2(x,y,z) &\coloneqq x \wedge z, \\
  J_3(x,y,z) &\coloneqq z \wedge (x \vee y).
\end{align*}
It is an easy exercise to show that these terms satisfy the identities above and
are hence J\'{o}nsson terms. It follows that $\AA$ is congruence distributive.
\end{proof} % }}}

Returning again to the fragment of Post's Lattice in
Figure~\ref{fig:PL_positive}, we are left to analyze the central diamond of
clones between $\Clone{AP}$ and $\Clone{A}$. The next lemma allows us to apply
Theorem~\ref{thm:simon} to these clones.

\begin{lem} \label{lem:AP_cong} % {{{
Let $\alg{B}$ be a 2-element algebra such that $\Clone{AP} = \alg{B}$  and let
$\AA = \alg{B}^n$. If $\theta$ is a congruence of $\AA$ then $\theta$ is also a
congruence of the group $\ZZ_2^n$.
\end{lem}
\begin{proof}
The clone $\Clone{AP}$ has ternary addition, $x+y+z$, as its generating
operation (see Figure~\ref{fig:posts_lattice}). In order to show that $\theta$
is a congruence of $\ZZ_2^n$, it is sufficient to show that the set $\theta$ is
closed under binary addition. Let $(a,a'), (b,b')\in \theta$. We have
\[
  a + b
  = (a + b + 0)
  \;\theta\;
  (a' + b' + 0)
  = a' + b',
\]
and hence $(a+a', b+b') \in \theta$, as desired.
\end{proof} % }}}

We are now ready to prove Theorem~\ref{thm:AP_MPT_quantum}, which we started the
section with. The proof splits into two portions, depending on whether the
algebra in question has clone contained in one of the infinite ``wings'' or
whether it is above the bottom of the central diamond.

\begin{customthm}{\ref{thm:AP_MPT_quantum}}[redux] % {{{
Let $\alg{B}$ be a 2-element algebra such that one of
\begin{enumerate}
  \item $\Clone{MPT}_0^{\infty} \preceq \alg{B}$,
  \item $\Clone{MPT}_1^{\infty} \preceq \alg{B}$, or
  \item $\Clone{AP} \preceq \alg{B}$
\end{enumerate}
holds and let $\AA = \alg{B}^n$. The quantum circuit $\Simon_\phi(n)$ of
Definition~\ref{defn:simon_circuit} solves $\HKP(\AA)$ with probability $1 -
1/\tau$ in $\Theta(n + \lg(\tau))$ iterations.
\end{customthm}
\begin{proof}
Consider an instance of the $\HKP(\AA)$ with homomorphism $\phi: \AA \to
\alg{X}$. If $\Clone{MPT}_0^{\infty} \preceq \alg{B}$ or $\Clone{MPT}_1^{\infty}
\preceq \alg{B}$, then $\alg{B}$ is congruence distributive. By
Lemma~\ref{lem:classical_CD_projs}, $\ker(\phi)$ is therefore the kernel of a
projection. All such kernels are also congruences of $\ZZ_2^n$. If, on the other
hand, $\Clone{AP} \preceq \alg{B}$, then by Lemma~\ref{lem:AP_cong} and
Observation~\ref{obs:clones_and_hkp}, we have that the congruence $\ker(\phi)$
of $\AA$ is also a congruence of the group $\ZZ_2^n$.

Theorem~\ref{thm:simon} therefore applies in all cases, so $\ker(\phi)$ can be
calculated with probability $1 - 1/\tau$ in $\Theta(n + \lg(\tau))$ iterations
of the circuit.
\end{proof} % }}}
%----------------------------------------------------------------------------}}}1

\section{The HKP and Post's Lattice: Hardness} \label{sec:hkp_negative}  % {{{1
The previous section presented quantum and classical algorithms for the $\HKP$
on the infinite upper portion of Post's Lattice. In this section, we establish
quantum and classical hardness results for the remaining clones. The main
classical result is Theorem~\ref{thm:no_classical} and the main quantum result
is Theorem~\ref{thm:no_QA}, stated below.

\begin{thm} \label{thm:no_classical}  % {{{
Let $\alg{B}$ be a 2-element algebra such that one of
\begin{enumerate}
  \item $\alg{B} \preceq \CloneWedge$,
  \item $\alg{B} \preceq \CloneVee$, or
  \item $\alg{B} \preceq \Clone{A}$
\end{enumerate}
holds, and let $\AA = \alg{B}^n$. Any classical algorithm which solves
$\HKP(\AA)$ must make $\Omega(2^{n/2})$ queries to the oracle.
\end{thm} % }}}

\begin{thm} \label{thm:no_QA} % {{{
Let $\alg{B}$ be a 2-element algebra such that one of
\begin{enumerate}
  \item $\alg{B} \preceq \CloneWedge$,
  \item $\alg{B} \preceq \CloneVee$, or
  \item $\alg{B} \preceq \Clone{U}$
\end{enumerate}
holds, and let $\AA = \alg{B}^n$. Any quantum algorithm which solves $\HKP(\AA)$
must make $\Omega((1 + \epsilon)^n)$ queries to the oracle for $0 < \epsilon <
1$.
\end{thm} % }}}

The best way to summarize these two results is in terms of Post's Lattice, see
Figure~\ref{fig:PL_negative}. In particular, powers of algebras whose clones are
contained in the interval between $\Clone{AP}$ and $\Clone{A}$ have $\HKP$ which
exhibits super-polynomial speedup using a quantum algorithm. Clones which are
below $\Clone{U}$, $\CloneWedge$, or $\CloneVee$ have $\HKP$ for which no
polynomial-time quantum algorithm exists. Combined with the results of the
previous section, this provides a complete quantum and classical classification
of the complexity of the $\HKP$ for powers of 2-element algebras.

% fig:PL_negative {{{2
\begin{figure} \centering 
  % lattice PL, negative {{{3
\begin{tikzpicture}[font=\tiny, yscale=0.85, xscale=0.85]
  % central section nodes {{{
  \path (0, 0)    node[no_poly_node, clone] (F) {}
       +(-1, 0.5) node[no_poly_node, clone] (UP0) {}
       +(1, 0.5)  node[no_poly_node, clone] (UP1) {}
       ++(0,1)    node[no_poly_node, clone] (UM) {}
       ++(0,1)    node[no_poly_node, clone] (UD) {}
       ++(0,1)    node[no_poly_node, clone] (U) {}
       ++(0,1)    node[superpoly, clone] (AP) {}
        +(-1,0.5) node[superpoly, clone] (AP0) {}
        +(0,0.5)  node[superpoly, clone] (AD) {}
        +(1,0.5)  node[superpoly, clone] (AP1) {}
       ++(0,1)    node[superpoly, clone] (A) {}; % }}}
  % left wing nodes {{{
  \path (UP0)
      ++(0,1.5)     node[no_poly_node, clone] (MeetP1) {}
      ++(-1,-0.5)   node[no_poly_node, clone] (MeetP) {}
       +(0,1)       node[no_poly_node, clone] (Meet) {}
      ++(-1,0.5)    node[no_poly_node, clone] (MeetP0) {}; % }}}
  % right wing nodes {{{
  \path (UP1)
      ++(0,1.5)      node[no_poly_node, clone] (JoinP0) {}
      ++(1,-0.5)     node[no_poly_node, clone] (JoinP) {}
       +(0,1)        node[no_poly_node, clone] (Join) {}
      ++(1,0.5)      node[no_poly_node, clone] (JoinP1) {};  % }}}

  % central section edges  {{{
  \draw[misc] (F) to[curveR] (AP)
    (UP0) to[out=90-15,in=-90] (AP0)
    (UP1) to[out=90+15,in=-90] (AP1)
    (UD) to[curveR] (AD)
    (U) to[curveL] (A);
  \draw[no_poly_line] (F) -- (UP0) -- (UM) -- (UP1) -- (F)
    (F) -- (MeetP)
    (F) -- (JoinP)
    (F) to[curveR] (UD)
    (UP0) -- (MeetP0)
    (UP0) -- (JoinP0)
    (UP1) -- (JoinP1)
    (UP1) -- (MeetP1)
    (UM) to[curveL] (U)
    (UM) -- (Meet)
    (UM) -- (Join)
    (UD) -- (U);
  \draw[superpoly] (AP) -- (AP0) -- (A) -- (AP1) -- (AP) -- (AD) -- (A); % }}}
  % left wing edges {{{
  \draw[no_poly_line] (MeetP) -- (MeetP0) -- (Meet) -- (MeetP1) -- (MeetP);  % }}}
  % right wing edges {{{
  \draw[no_poly_line] (JoinP) -- (JoinP0) -- (Join) -- (JoinP1) -- (JoinP);  % }}}

  % wing clone labels
  \node[cloneLabel, above] at (Meet) {$\CloneWedge$};
  \node[cloneLabel, above] at (Join) {$\CloneVee$};

  % central clone labels
  \node[cloneLabel, left, yshift=-0.25em] at (F) {$\CloneBot$};
  \node[cloneLabel, left] at (U) {$\Clone{U}$};
  \node[cloneLabel, above, xshift=0.3em] at (A) {$\Clone{A}$};
  \node[cloneLabel, left, align=left] at (AP0) {$\Clone{AP}_0$ \\ $\approx \ZZ_2$};
  \node[cloneLabel, below] at (AP) {$\Clone{AP}$};

  % key
  \coordinate (key_center) at ($(current bounding box.north east)!0.5!(current bounding box.south west)$);
  \coordinate (superpoly_key) at ($(key_center)+(-3.45,-3.5)$);
  \draw[superpoly] ($(superpoly_key.west)+(-0.215,-0.25)$) -- node[black, right]{:} ($(superpoly_key.west)+(-0.215,0.25)$);
  \node[anchor=west, align=left] at (superpoly_key) {only quantum poly-time solution \\ (super-polynomial speedup)};
  \coordinate (no_poly_key) at ($(key_center)+(0.85,-3.5)$);
  \draw[no_poly_line, very thick] ($(no_poly_key.west)+(-0.215,-0.25)$) -- node[black, right]{:} ($(no_poly_key.west)+(-0.215,0.25)$);
  \node[anchor=west, align=left] at (no_poly_key) {no poly-time classical \\ or quantum solution exists};
\end{tikzpicture} % }}}3
  \caption{ The fragment of Post's Lattice for which no classical
    polynomial-time algorithm for the $\HKP$ exists (whole diagram). Dotted
    lines indicate clones for which no polynomial-time quantum algorithm exists,
    while bold lines indicate clones for which no polynomial-time classical
    algorithm exists, but a polynomial-time quantum algorithm does exist (i.e.\
    those clones which exhibit super-polynomial speedup). }
  \Description{ The fragment of Post's Lattice consisting of clones which are
    contained in $\CloneWedge$, $\CloneVee$, or $\Clone{A}$. The interval
    $[\Clone{AP}, \Clone{A}]$ exhibits superpolynomial-speedup with a quantum
    algorithm. }
  \label{fig:PL_negative}
\end{figure}
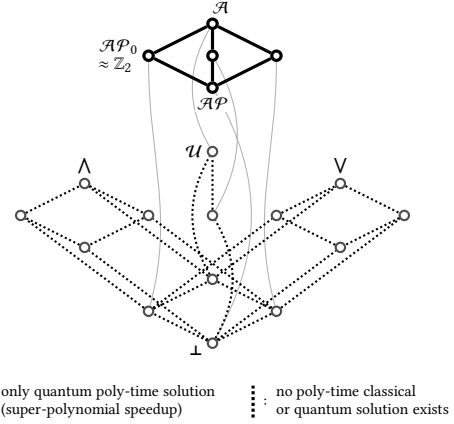  % }}}2

The remainder of the section is devoted to proving
Theorems~\ref{thm:no_classical} and~\ref{thm:no_QA}. The proof technique relies
on various counting arguments and a close analysis of congruence generation in
different algebras. The primary tool we use is a result due to
Maltsev~\cite{Maltsev_CongChain}, commonly called a ``Maltsev chain''.

\begin{defn} \label{defn:maltsev_chain} % {{{
Let $\AA$ be an algebra and $X\subseteq A^2$. If $\alpha, \beta\in A$ are such
that there exists an $m$-ary term operation $t$ and pairs $(a_1,b_1), \dots,
(a_m, b_m)\in X\cup \Cong0$ with 
\[
  \big\{ \alpha, \beta \big\} 
  = \big\{ t(a_1, \dots, a_m), t(b_1, \dots, b_m) \big\}
\]
then we will write $\alpha \leadsto \beta$. A \emph{Maltsev chain}
between $\alpha$ and $\beta$ is a sequence $\lambda_0, \dots, \lambda_n \in A$
such that
\[
  \alpha
  = \lambda_0
  \leadsto \lambda_1
  \leadsto 
  \dots
  \leadsto \lambda_n
  = \beta,
\]
written $\alpha \leadsto^n \beta$.
\end{defn}  % }}}

\begin{prop}[Maltsev~\cite{Maltsev_CongChain}] \label{prop:matlsev_chain}  % {{{
Let $\AA$ be an algebra, $X\subseteq A^2$, and $\theta = \Cg(X)$. Then
$(\alpha,\beta)\in \theta$ if and only if $\alpha \leadsto^n \beta$
for some $n$.
\end{prop}  % }}}

We begin by proving that if $\alg{B} \preceq \CloneWedge$ or $\alg{B} \preceq
\CloneVee$, then exponentially-many queries to the oracle are necessary to solve
the $\HKP$ for powers of $\alg{B}$ for both classical and quantum algorithms.
This proves items~1 and~2 of Theorems~\ref{thm:no_classical}
and~\ref{thm:no_QA}. The proof of this is contained in the following lemma,
combined with Observation~\ref{obs:clones_and_hkp}.

\begin{lem} \label{lem:vee_wedge_no_P_alg} % {{{
Let $\alg{B}$ be a 2-element algebra such that $\Clo(\alg{B}) = \CloneWedge$ or
$\Clo(\alg{B}) = \CloneVee$ and let $\AA = \alg{B}^n$. Any algorithm (classical or
quantum) solving $\HKP(\AA)$ must make $\Omega((1 + \epsilon)^n)$ queries to the
oracle for $0 < \epsilon < 1$.
\end{lem}
\begin{proof}
We will present the argument for $\Clo(\alg{B}) = \CloneWedge$. The argument for
$\Clo(\alg{B}) = \CloneVee$ is analogous. The clone $\CloneWedge$ has $x\wedge
y$ amongst its generating operations (see Figure~\ref{fig:posts_lattice}). The
set $\{0,1\}$ has a natural order (namely $0\leq 1$), and this can be extended
to an ordering on $\{0,1\}^n$. The operation of $\wedge$ is compatible with this
order: $x \wedge y = x$ if and only if $x \leq y$.

Define a set of $\leq$-incomparable elements of $\{0,1\}^n$,
\[
  M
  = \big\{ a \in A \mid \text{exactly } \lfloor n/2 \rfloor 
    \text{ coordinates of $a$ are $0$} \big\}
\]
and choose arbitrary $Z\subseteq M$. Let $\theta_Z$ be the smallest congruence
of $\AA$ with $Z^2$ contained in a single equivalence class --- that is,
$\theta_Z \coloneqq \Cg(Z^2)$. We will carefully analyze the generation of
$\theta_Z$ via a series of claims.

\begin{claim*}
Let $X\subseteq A^2$ and $\theta = \Cg(X)$. If $\alpha \;\theta\; \beta$ and
$\alpha \neq \beta$ then there is $(a,b)\in X$ with $\alpha \leq a$ or
$\alpha\leq b$.
\end{claim*}
\begin{claimproof}
By Proposition~\ref{prop:matlsev_chain}, there is an $m$-ary term operation $t$
and pairs $(a_1,b_1), \dots, (a_m,b_m)\in X\cup \Cong0$ such that $\alpha \in \{
t(a_1, \dots, a_m), t(b_1, \dots, b_m) \}$. Each term operation of $\alg{B}$
(and hence $\AA$) can be written
\[
  t(x_1, \dots, x_m) = x_{i_1} \wedge \dots \wedge x_{i_k}
\]
where $\{i_1, \dots, i_k\} \subseteq [m]$. The claim follows.
\end{claimproof}

\begin{claim*}
If $z\in Z$ and $\gamma > z$ then $(z,\gamma) \not\in\theta_Z$.
\end{claim*}
\begin{claimproof}
Suppose towards a contradiction that $z \;\theta_Z\; \gamma$. By the first
claim, there must be some $z'\in Z$ such that $\gamma\leq z'$. This implies that
$z < z'$, but since $Z\subseteq M$ and $M$ consists of $\leq$-incomparable
elements, this is impossible.
\end{claimproof}

\begin{claim*}
If $X \subseteq A^2$ is such that $\theta_Z = \Cg(X)$, then $|Z| \leq |X|$.
\end{claim*}
\begin{claimproof}
By the first claim, for each $z\in Z$ there is some $(a_z, b_z)\in X$ with
(without loss of generality) $z\leq a_z$. We argue that each $a_z$ is distinct
and thus $|Z|\leq |X|$.

Suppose that there are distinct $z, z'\in Z$ such that $a_z = a_{z'}$.
Summarizing our assumptions,
\[
  (a_z, b_z) 
  \in X
  \subseteq \Cg(X) 
  = \theta_Z
  = \Cg(Z^2).
\]
Applying the first claim with $Z$ in place of $X$ yields some $u\in Z$ such that
$b_z \leq u$. Since $(a_z, b_z), (z, u)\in \theta_Z$, we have
\[
  z
  = \big( a_z \wedge z \big)
  \;\theta_Z\; \big( b_z \wedge u \big)
  = b_z
\]
and hence $a_z \;\theta_Z\; z$. The second claim together with $z\leq a_z$ now
yields a contradiction unless $z = a_z$. The same argument for $z'$ then gives
us $z = a_z = a_{z'} = z'$, contradicting $z$ and $z'$ being distinct and
finishing the proof of the claim.
\end{claimproof}

From the claim above, determining the congruence $\theta_Z$ requires specifying
at least $|Z|$-many elements. Using Stirling's approximation, we have
\[
  |M|
  = \binom{n}{n/2}
  = \frac{n!}{\big( (n/2)! \big)^2}
  \approx \frac{\sqrt{n} (n/e)^n}{ \big( \sqrt{n} (n/(2e))^{n/2} \big)^2}
  = \frac{2^n}{\sqrt{n}},
\]
where the ``$\approx$'' symbol means that both sides have the same
$\Theta$-complexity class. There are thus exponentially-many subsets $Z$ with
$|Z|\in \Omega((1 + \epsilon)^n)$ for $0 < \epsilon < 1$. Any algorithm (quantum
or classical) for solving $\HKP(\AA)$ must distinguish between these subsets,
and by the previous claim this requires specifying at least $|Z|$-many elements.
\end{proof} % }}}

In order to complete the proof of Theorem~\ref{thm:no_classical}, which we
opened the section with, we need only prove item~3: for all $\alg{B} \preceq
\Clone{A}$, the $\HKP$ for powers of $\alg{B}$ \emph{classically} requires
exponentially-many oracle queries. The proof of this is contained in the
following lemma, combined with Observation~\ref{obs:clones_and_hkp}.

\begin{lem} \label{lem:A_no_classical} % {{{
Let $\alg{B}$ be a 2-element algebra with $\Clo(\alg{B}) = \Clone{A}$ and let
$\AA = \alg{B}^n$. Any classical (resp.\ probabilistic) algorithm solving
$\HKP(\AA)$ must make $\Omega(2^n)$ (resp.\ $\Omega(2^{n/2})$) queries to the
oracle.
\end{lem}
\begin{proof}
The clone $\Clone{A}$ has the operation $x \Aiff y$ as a generating operation
(see Figure~\ref{fig:posts_lattice}). Choose a random $z\in A \setminus \{ 0^n
\}$ and define $\theta_z$ to be the congruence generated by relating $0^n$ and
$z$,
\[
  \theta_z \coloneqq \Cg(0^n, z).
\]
Observe that modulo-2 addition is definable in $\Clone{A}$ (and hence $\alg{B}$
and $\AA$) by
\[
  x + y \coloneqq
  x \Aiff (0 \Aiff y).
\]

\begin{claim*}
The congruence $\theta_z$ has equivalence classes $x/\theta_z \coloneqq \big\{
x, x+z \big\}$.
\end{claim*}
\begin{claimproof}
Take $\psi$ to be the equivalence relation with the claimed equivalence
classes. To prove that $\psi$ is a congruence we need only show that if $a
\;\psi\; a'$ and $b \;\psi\; b'$ then
\[
  ( a \Aiff b) \;\psi\; ( a' \Aiff b' ).
\]
By the definition of $\psi$, we have that $a' = a + \epsilon z$ and $b' = b +
\tau z$ for some $\epsilon, \tau\in \{0, 1\}$. The operation $\Aiff$ is
commutative and associative, so using ``$+$'' as defined before the claim, we
have
\begin{align*}
  (w + x) \Aiff y 
    &= w \Aiff 0 \Aiff x \Aiff y \\
  &= w \Aiff 0 \Aiff y \Aiff x 
    = (w + y) \Aiff x.
\end{align*}
Thus $(w+x) \Aiff y = (w+y) \Aiff x$. Using this identity twice,
\begin{align*}
  a' \Aiff b'
    &= (a + \epsilon z) \Aiff (b + \tau z) \\
  &= \big( a \Aiff (b + \tau z) \big) + \epsilon z
    = ( a \Aiff b ) + (\epsilon + \tau) z.
\end{align*}
Addition is defined modulo 2, so $\epsilon + \tau\in \{0,1\}$ and we have $( a
\Aiff a) \;\psi\; ( a' \Aiff b' )$, as desired. The congruence $\theta_z$ is
minimal and contains $\psi$, so $\theta_z = \psi$.
\end{claimproof}

We will now show that it is exponentially hard for a classical algorithm to
distinguish between $\theta_z$ and the identity congruence $\Cong0$ for randomly
chosen $z$. Let $\phi$ be the homomorphism
\[ \begin{array}{rcl}
  \phi : \AA & \to     & \AA/\theta_z, \\
           a & \mapsto & a/\theta_z,
\end{array} \]
so that $\ker(\phi) = \theta_z$. 

Suppose that we have a classical procedure for solving $\HKP(\AA)$, and that
this procedure evaluates $\phi$ on a subset $E = \{e_1, \dots, e_\ell\}
\subseteq A$. We are able to distinguish $\theta_z$ from $\Cong0$ if and only if
$|\phi(E)| < |E|$. By the claim above, for distinct $e_i, e_j\in E$, we have
that $\phi(e_i) = \phi(e_j)$ if and only if $e_i = e_j + z$. The element $z$ was
chosen randomly, so the probability of this occurring is $1/(|\AA|-1)$. It
follows that for fixed $j$, the probability that $e_i = e_j + z$ for some $i<j$
is
\[
  \Prob\Big(\exists i < j \; \big[ e_i = e_j + z \big] \Big)
  = \frac{j-1}{|\AA| - 1}.
\]
Therefore the probability that for some distinct $e_i, e_j\in E$ we have $e_i =
e_j + z$ (equivalently, $|\phi(E)| < |E|$) is
\begin{align*}
  &\Prob\big( |\phi(E)| < |E| \big)
  = \Prob\Big(\exists i, j \; \big[ e_i = e_j + z \big] \Big) \\
  &\qquad = \sum_{j=1}^{|E|} \Prob\Big(\exists i < j \; \big[ e_i = e_j + z \big] \Big)
  = \sum_{j=1}^{|E|} \frac{j-1}{|\AA| - 1} \\
  &\qquad = \frac{|E|(|E|-1)}{2(|\AA| - 1)}.
\end{align*}
Thus, the only way for a classical algorithm to correctly distinguish $\theta_z$
from $\Cong0$ with probability at least $1/2$ is by making $|E| \in
\Omega(|\AA|^{1/2}) = \Omega(2^{n/2})$ evaluations of $\phi$. To be correct with
probability 1 requires $|E|\in \Omega(2^n)$ evaluations.
\end{proof} % }}}

We now complete the proof of Theorem~\ref{thm:no_QA} from the start of the
section. It remains to prove the last item of that theorem, that for all
$\alg{B} \preceq \Clone{U}$, the $\HKP$ for powers of $\alg{B}$ requires
exponentially-many oracle queries for a \emph{quantum} algorithm.
Observation~\ref{obs:clones_and_hkp} together with the next lemma establish
this.

\begin{lem} \label{lem:U_no_quantum}  % {{{
Let $\alg{B}$ be a 2-element algebra with $\Clo(\alg{B}) = \Clone{U}$ and let
$\AA = \alg{B}^n$. Any algorithm (classical or quantum) solving $\HKP(\AA)$ must
make $\Omega(2^n)$ queries to the oracle.
\end{lem}
\begin{proof}
The clone $\Clone{U}$ has $\neg x$ amongst its generating operations (see
Figure~\ref{fig:posts_lattice}). Partition $A$ into four disjoint sets of equal
size, $A_0$, $B_0$, $A_1$, $B_1$, such that
\[
  A_1
    = \neg A_0 
    = \big\{ \neg a \mid a\in A_0 \big\},
  \qquad
  B_1 
    = \neg B_0 
    = \big\{ \neg b \mid b\in B_0 \big\}.
\]
Enumerate the elements of $A_0$ and $B_0$ as $a_1, \dots, a_{2^{n-2}}$ and
$b_1, \dots, b_{2^{n-2}}$ respectively. Define $C = \{ (a_i, b_i) \mid i\in
[2^{n-2}] \}$. Finally, for each subset $Y\subseteq C$ let $\theta_Y = \Cg(Y)$.

We will show that all the $\theta_Y$ are distinct and that specifying $\theta_Y$
by generators requires finding a set of size $|Y|$. There are $\Omega(2^n)$-many
distinct $Y$ of size $\Omega(2^n)$, so this is sufficient to prove the lemma.

\begin{claim*}
Suppose that $Y = \{ (a_i, b_i) \mid i\in I \}\subseteq C$ for some $I\subseteq
[2^{n-2}]$. Then
\[
  \theta_Y
  = \Cong0 \cup Y \cup Y^\partial \cup (\neg Y) \cup (\neg Y)^\partial
\]
where $\neg Y \coloneqq \{ (\neg a_i, \neg b_i) \mid i\in I \}$ and $Z^\partial
\coloneqq \{ (y,x) \mid (x,y)\in Z \}$ for $Z\subseteq A^2$.
\end{claim*}
\begin{claimproof}
Let $\psi = \Cong0 \cup Y \cup Y^\partial \cup (\neg Y) \cup (\neg Y)^\partial$.
The congruence $\theta_Y$ is the least congruence containing $Y$, so $\psi
\subseteq \theta_Y$. In order to show that $\psi = \theta_Y$, it is therefore
enough to show that $\psi$ is closed under the operation of $\neg$. This is
clear from the construction of $\psi$.
\end{claimproof}

By the above claim, each distinct $Y\subseteq C$ determines a distinct
$\theta_Y$, and specifying $\theta_Y$ requires producing a subset of size $|Y|$.
Since there are $\Omega(2^n)$-many distinct $Y$ of size $\Omega(2^n)$, the
conclusion of the lemma follows.
\end{proof} % }}}
%----------------------------------------------------------------------------}}}1

\section{Conclusion} \label{sec:concl}  % {{{1
% fig:concl_lattice {{{2
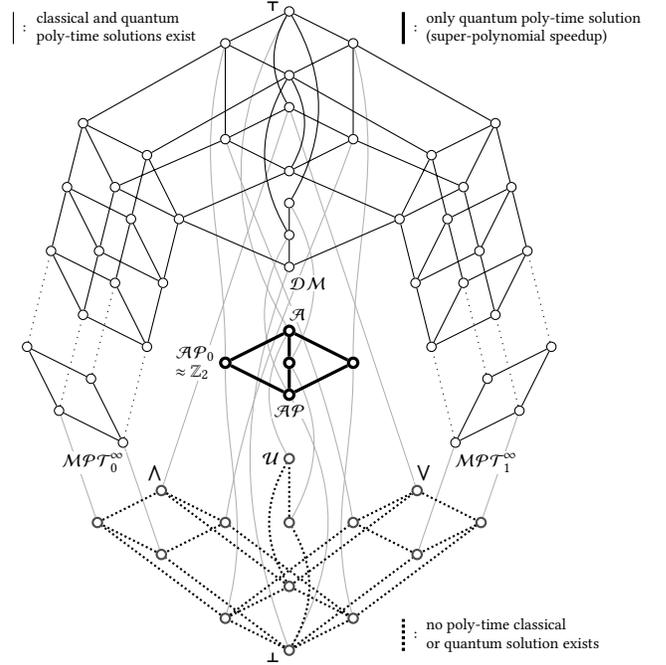
\begin{figure} \centering 
% Post's lattice {{{3
\begin{tikzpicture}[font=\tiny, yscale=0.85, xscale=0.85]
  % central section nodes {{{
  \path (0, 0)    node[no_poly_node, clone] (F) {}
       +(-1, 0.5) node[no_poly_node, clone] (UP0) {}
       +(1, 0.5)  node[no_poly_node, clone] (UP1) {}
       ++(0,1)    node[no_poly_node, clone] (UM) {}
       ++(0,1)    node[no_poly_node, clone] (UD) {}
       ++(0,1)    node[no_poly_node, clone] (U) {}
       ++(0,1)    node[superpoly, clone] (AP) {}
        +(-1,0.5) node[superpoly, clone] (AP0) {}
        +(0,0.5)  node[superpoly, clone] (AD) {}
        +(1,0.5)  node[superpoly, clone] (AP1) {}
       ++(0,1)    node[superpoly, clone] (A) {}
       ++(0,1)    node[class_and_quant, clone] (DM) {}
       ++(0,0.5)  node[class_and_quant, clone] (DP) {}
       ++(0,0.5)  node[class_and_quant, clone] (D) {}
       ++(0,0.5)  node[class_and_quant, clone] (MP) {}
        +(-1,0.5) node[class_and_quant, clone] (MP0) {}
        +(1,0.5)  node[class_and_quant, clone] (MP1) {}
       ++(0,1)    node[class_and_quant, clone] (M) {}
       ++(0,0.5)  node[class_and_quant, clone] (P) {}
        +(-1,0.5) node[class_and_quant, clone] (P0) {}
        +(1,0.5)  node[class_and_quant, clone] (P1) {}
       ++(0,1)    node[class_and_quant, clone] (T) {}; % }}}
  % left wing nodes {{{
  \path (UP0)
      ++(0,1.5)     node[no_poly_node, clone] (MeetP1) {}
      ++(-1,-0.5)   node[no_poly_node, clone] (MeetP) {}
       +(0,1)       node[no_poly_node, clone] (Meet) {}
      ++(-1,0.5)    node[no_poly_node, clone] (MeetP0) {}

      ++(0.4,1.25)  node[class_and_quant, clone] (MPT0inf) {}
       +(-0.5,1)    node[class_and_quant, clone] (PT0inf) {}
      ++(-1,0.5)    node[class_and_quant, clone] (MT0inf) {}
       +(-0.5,1)    node[class_and_quant, clone] (T0inf) {}
        (MPT0inf)
      ++(1.5/4,1.5) node[class_and_quant, clone] (MPT04) {}
       +(-0.5,1)    node[class_and_quant, clone] (PT04) {}
      ++(-1,0.5)    node[class_and_quant, clone] (MT04) {}
       +(-0.5,1)    node[class_and_quant, clone] (T04) {}
       (MPT04)
      ++(0.25,1)    node[class_and_quant, clone] (MPT03) {}
       +(-0.5,1)    node[class_and_quant, clone] (PT03) {}
      ++(-1,0.5)    node[class_and_quant, clone] (MT03) {}
       +(-0.5,1)    node[class_and_quant, clone] (T03) {}
       (MPT03)
      ++(0.25,1)    node[class_and_quant, clone] (MPT02) {}
       +(-0.5,1)    node[class_and_quant, clone] (PT02) {}
      ++(-1,0.5)    node[class_and_quant, clone] (MT02) {}
       +(-0.5,1)    node[class_and_quant, clone] (T02) {}; % }}}
  % right wing nodes {{{
  \path (UP1)
      ++(0,1.5)      node[no_poly_node, clone] (JoinP0) {}
      ++(1,-0.5)     node[no_poly_node, clone] (JoinP) {}
       +(0,1)        node[no_poly_node, clone] (Join) {}
      ++(1,0.5)      node[no_poly_node, clone] (JoinP1) {}

      ++(-0.4,1.25)  node[class_and_quant, clone] (MPT1inf) {}
       +(0.5,1)      node[class_and_quant, clone] (PT1inf) {}
      ++(1,0.5)      node[class_and_quant, clone] (MT1inf) {}
       +(0.5,1)      node[class_and_quant, clone] (T1inf) {}
        (MPT1inf)
      ++(-1.5/4,1.5) node[class_and_quant, clone] (MPT14) {}
       +(0.5,1)      node[class_and_quant, clone] (PT14) {}
      ++(1,0.5)      node[class_and_quant, clone] (MT14) {}
       +(0.5,1)      node[class_and_quant, clone] (T14) {}
       (MPT14)
      ++(-0.25,1)    node[class_and_quant, clone] (MPT13) {}
       +(0.5,1)      node[class_and_quant, clone] (PT13) {}
      ++(1,0.5)      node[class_and_quant, clone] (MT13) {}
       +(0.5,1)      node[class_and_quant, clone] (T13) {}
       (MPT13)
      ++(-0.25,1)    node[class_and_quant, clone] (MPT12) {}
       +(0.5,1)      node[class_and_quant, clone] (PT12) {}
      ++(1,0.5)      node[class_and_quant, clone] (MT12) {}
       +(0.5,1)      node[class_and_quant, clone] (T12) {};  % }}}

   % central section edges (misc)  {{{
   \draw[misc] (F) to[curveR] (AP)
     (F) to[out=90+20,in=-90-20] (DM)
     (UP0) to[out=90-15,in=-90] (AP0)
     (UP1) to[out=90+15,in=-90] (AP1)
     (UD) to[curveR] (AD)
     (U) to[curveL] (A)
     (AP) to[curveL] (DP)
     (AP0) to[out=90,in=-90-15] (P0)
     (AP1) to[out=90,in=-90+15] (P1)
     (AD) to[curveR] (D)
     (A) to[curveL] (T);  % }}}
   % left wing edges (misc) {{{
   \draw[misc] (MeetP1) to[out=90-10,in=-90-20] (MP1)
     (MeetP) -- (MPT0inf)
     (MeetP0) -- (MT0inf)
     (Meet) -- (M);  % }}}
  % right wing edges (misc) {{{
  \draw[misc] (JoinP0) to[out=90+10,in=-90+20] (MP0)
    (JoinP) -- (MPT1inf)
    (JoinP1) -- (MT1inf)
    (Join) -- (M);  % }}}
  % central section edges  {{{
  \draw[no_poly_line] (F) -- (UP0) -- (UM) -- (UP1) -- (F)
    (F) -- (MeetP)
    (F) -- (JoinP)
    (F) to[curveR] (UD)
    (UP0) -- (MeetP0)
    (UP0) -- (JoinP0)
    (UP1) -- (JoinP1)
    (UP1) -- (MeetP1)
    (UM) to[curveL] (U)
    (UM) -- (Meet)
    (UM) -- (Join)
    (UD) -- (U);
  \draw[superpoly] (AP) -- (AP0) -- (A) -- (AP1) -- (AP) -- (AD) -- (A);
  \draw[class_and_quant] (DM) -- (MPT02) -- (MP) -- (MPT12) -- (DM) -- (DP) -- (D)
    (DP) to[curveL] (P)
    (D) to[curveR] (T)
    (MP) -- (MP0) -- (M) -- (MP1) -- (MP)
    (MP) to[curveR] (P)
    (MP0) -- (P0)
    (MP1) -- (P1)
    (M) to[curveL] (T)
    (P) -- (P0) -- (T) -- (P1) -- (P); % }}}
  % left wing edges {{{
  \draw[no_poly_line] (MeetP) -- (MeetP0) -- (Meet) -- (MeetP1) -- (MeetP);
  \draw[class_and_quant] (MPT0inf) -- (MT0inf) -- (T0inf) -- (PT0inf) -- (MPT0inf)
    (MPT0inf) edge[dotted] (MPT04)
    (MT0inf) edge[dotted] (MT04)
    (PT0inf) edge[dotted] (PT04)
    (T0inf) edge[dotted] (T04)
    (MPT04) -- (MT04) -- (T04) -- (PT04) -- (MPT04)
    (MPT04) -- (MPT03) (MT04) -- (MT03) (PT04) -- (PT03) (T04) -- (T03)
    (MPT03) -- (MT03) -- (T03) -- (PT03) -- (MPT03)
    (MPT03) -- (MPT02) (MT03) -- (MT02) (PT03) -- (PT02) (T03) -- (T02)
    (MPT02) -- (MT02) -- (T02) -- (PT02) -- (MPT02)
    (MT02) -- (MP0)
    (T02) -- (P0)
    (PT02) -- (P); % }}}
  % right wing edges {{{
  \draw[no_poly_line] (JoinP) -- (JoinP0) -- (Join) -- (JoinP1) -- (JoinP);
  \draw[class_and_quant] (MPT1inf) -- (MT1inf) -- (T1inf) -- (PT1inf) -- (MPT1inf)
    (MPT1inf) edge[dotted] (MPT14)
    (MT1inf) edge[dotted] (MT14)
    (PT1inf) edge[dotted] (PT14)
    (T1inf) edge[dotted] (T14)
    (MPT14) -- (MT14) -- (T14) -- (PT14) -- (MPT14)
    (MPT14) -- (MPT13) (MT14) -- (MT13) (PT14) -- (PT13) (T14) -- (T13)
    (MPT13) -- (MT13) -- (T13) -- (PT13) -- (MPT13)
    (MPT13) -- (MPT12) (MT13) -- (MT12) (PT13) -- (PT12) (T13) -- (T12)
    (MPT12) -- (MT12) -- (T12) -- (PT12) -- (MPT12)
    (MT12) -- (MP1)
    (T12) -- (P1)
    (PT12) -- (P); % }}}

  % wing clone labels
  \node[cloneLabel, above, xshift=-0.25em] at (Meet) {$\CloneWedge$};
  \node[cloneLabel, above, xshift=0.25em] at (Join) {$\CloneVee$};
  \node[cloneLabel, below left, xshift=0.43em] at (MPT0inf) {$\Clone{MPT}_0^\infty$};
  \node[cloneLabel, below right, xshift=-0.43em] at (MPT1inf) {$\Clone{MPT}_1^\infty$};
  \node[cloneLabel, left, yshift=0.2em] at (T) {$\CloneTop$};

  % central clone labels
  \node[cloneLabel, left, yshift=-0.25em] at (F) {$\CloneBot$};
  \node[cloneLabel, left] at (U) {$\Clone{U}$};
  \node[cloneLabel, above, xshift=0.3em] at (A) {$\Clone{A}$};
  \node[cloneLabel, left, align=left] at (AP0) {$\Clone{AP}_0$ \\ $\approx \ZZ_2$};
  \node[cloneLabel, below] at (AP) {$\Clone{AP}$};
  \node[cloneLabel, below right, xshift=-0.4em] at (DM) {$\Clone{DM}$};

  % key
  \coordinate (key_center) at ($(current bounding box.north east)!0.5!(current bounding box.south west)$);
  \coordinate (class_and_quant_key) at ($(key_center)+(-4.1,4.75)$);
  \draw[class_and_quant] ($(class_and_quant_key.west)+(-0.215,-0.25)$) -- node[black, right]{:} ($(class_and_quant_key.west)+(-0.215,0.25)$);
  \node[anchor=west, align=left] at (class_and_quant_key) {classical and quantum \\ poly-time solutions exist};
  \coordinate (superpoly_key) at ($(key_center)+(2,4.75)$);
  \draw[superpoly] ($(superpoly_key.west)+(-0.215,-0.25)$) -- node[black, right]{:} ($(superpoly_key.west)+(-0.215,0.25)$);
  \node[anchor=west, align=left] at (superpoly_key) {only quantum poly-time solution \\ (super-polynomial speedup)};
  \coordinate (no_poly_key) at ($(key_center)+(2,-4.75)$);
  \draw[no_poly_line, very thick] ($(no_poly_key.west)+(-0.215,-0.25)$) -- node[black, right]{:} ($(no_poly_key.west)+(-0.215,0.25)$);
  \node[anchor=west, align=left] at (no_poly_key) {no poly-time classical \\ or quantum solution exists};
\end{tikzpicture} % }}}3
  \caption{ Diagram of Post's Lattice. The $\HKP$ for powers of clones in bold
    have a efficient quantum solution, but no efficient classical solution
    (i.e.\ they exhibit super-polynomial speedup). Powers of clones in plain
    lines have polynomial-time classical and quantum solutions to their $\HKP$s,
    while powers of clones indicated with dotted lines have neither classical
    nor quantum polynomial-time solutions to their respective $\HKP$s. The clone
    of the group $\ZZ_2$ is indicated.}
  \Description{ Figures~\ref{fig:PL_positive} and \ref{fig:PL_negative}
    combined. }
  \label{fig:concl_lattice}
\end{figure}  % }}}2

Combining Theorem~\ref{thm:AP_MPT_quantum}, Corollary~\ref{cor:classical_CD},
Theorem~\ref{thm:no_classical}, and Theorem~\ref{thm:no_QA} provides a complete
classification of the quantum and classical algorithmic complexity of
$\HKP(\alg{B}^n)$, where $\alg{B}$ is a 2-element algebra. This classification
is summarized in the theorem below and in Figure~\ref{fig:concl_lattice}.

\begin{thm} \label{thm:concl_classification} % {{{
Let $\alg{B}$ be a 2-element algebra.
\begin{enumerate}
  \item If $\Clone{MPT}_0^{\infty} \preceq \alg{B}$, $\Clone{MPT}_1^{\infty}
    \preceq \alg{B}$, or $\Clone{DM} \preceq \alg{B}$, then there exist both
    classical and quantum polynomial-time algorithms for $\HKP(\alg{B}^n)$.

  \item If $\Clone{AP} \preceq \alg{B} \preceq \Clone{A}$, then a quantum
    polynomial-time algorithm solving $\HKP(\alg{B}^n)$ exists. Furthermore, no
    classical polynomial-time algorithm for $\HKP(\alg{B}^n)$ exists.

  \item If $\alg{B} \preceq \CloneWedge$, $\alg{B} \preceq \CloneVee$, or
  $\alg{B} \preceq \Clone{U}$ then no quantum or classical polynomial-time
  algorithm for $\HKP(\alg{B}^n)$ exists.
\end{enumerate}
\end{thm} % }}}

This classification can be seen as a broad extension of the results of
Simon~\cite{Simon_Alg}, which are included in item~2 of the theorem when
$\alg{B} = \ZZ_2$.

There are many open questions surrounding the Hidden Kernel Problem, a few of
which we detail below. Theorems~\ref{thm:concl_classification}
and~\ref{thm:no_classical} state that both classically and using a quantum
algorithm, $\HKP(\AA)$ is either in \comp{P} or \comp{EXP} when $\AA$ is the
power of a 2-element algebra.

\begin{question*}
For \emph{arbitrary} fixed finite algebras $\AA$, precisely which complexity
classes (quantum or classical) are parameterized by $\HKP(\AA)$?
\end{question*}

Turning this problem around, we can instead ask for algebraic properties which
enforce the existence of an efficient quantum solution to $\HKP(\AA^n)$.

\begin{question*}
Consider $\HKP(\AA^n)$, where $\AA$ is an arbitrary finite algebra. What
structural conditions on $\AA$ ensure that $\HKP(\AA^n)$ always admits a
polynomial-time quantum solution?
\end{question*}

In an algebra $\AA$, the term operation $t(x_1, \dots, x_n)$ is said to be a
\emph{cube term} if for every $i \in [n]$ there is a choice of $u_1, \dots,
u_n\in \{x,y\}$ with $u_i = y$ such that the equation $t(u_1, \dots, u_n) = x$
holds in $\AA$. The clone $\Clone{AP}$ has a cube term given by $x + y + z$ and
$\Clone{DM}$ has a cube term given by $\text{maj}(x,y,z)$. More generally, every
group has a cube term given by $xy^{-1}z$.

The study of cube terms originated with algebraic approach to the
\emph{Constraint Satisfaction Problem (CSP)} in~\cite{BIMMVW_FewSubalgs}, and
the existence of a cube term for $\AA$ is associated with some quite strong
regularity conditions on the structure of powers of $\AA$.

If $\AA$ has a cube term, then subalgebras of $\AA^n$ have a generating set
which is bounded by a polynomial in $n$. This applies in particular to
congruences of $\AA^n$. It follows that if $\AA$ has a cube term, then counting
arguments similar to those in the proofs of Lemmas~\ref{lem:vee_wedge_no_P_alg}
and~\ref{lem:U_no_quantum} will not be sufficient to rule out the existence of a
quantum algorithm. This leads us to conjecture the following.

\begin{conj*}
If $\AA$ has a cube term, then $\HKP(\AA^n)$ has an efficient quantum solution.
\end{conj*}

As mentioned above, every group has a cube term. An efficient quantum solution
to the above conjecture would therefore restrict to an efficient quantum
solution to the hidden normal subgroup problem.
%----------------------------------------------------------------------------}}}1

\balance

\bibliographystyle{ACM-Reference-Format}   % {{{1
\bibliography{hkp-post,references}
%----------------------------------------------------------------------------}}}1
\end{document}